\documentclass[11pt,a4paper]{amsart}
\usepackage{amssymb, amstext, amscd, amsmath, color}
\usepackage{graphicx}
\usepackage{url}
\usepackage{tikz,enumerate}
\usepackage{pgfplots}
\usepackage{url}
\usepackage{kbordermatrix} 
\renewcommand{\kbldelim}{[} 
\renewcommand{\kbrdelim}{]}
\usepackage{hhline}
\usepackage{hyperref}

\pgfplotsset{compat=1.10}

\begin{document}

\newtheorem{thm}{Theorem}[section]
\newtheorem{cor}[thm]{Corollary}
\newtheorem{prop}[thm]{Proposition}
\newtheorem{lem}[thm]{Lemma}
%
\theoremstyle{definition}
\newtheorem{rem}[thm]{Remark}
\newtheorem{defn}[thm]{Definition}
\newtheorem{note}[thm]{Note}
\newtheorem{eg}[thm]{Example}
\newcommand{\Prf}{\noindent\textbf{Proof.\ }}
\newcommand{\bx}{\hfill$\blacksquare$\medbreak}
\newcommand{\upbx}{\vspace{-2.5\baselineskip}\newline\hbox{}%
\hfill$\blacksquare$\newline\medbreak}
\newcommand{\eqbx}[1]{\medbreak\hfill\(\displaystyle #1\)\bx}

\newcommand{\FFock}{\mathcal{F}}
\newcommand{\kil}{\mathsf{k}}
\newcommand{\Hil}{\mathsf{H}}
\newcommand{\hil}{\mathsf{h}}
\newcommand{\Kil}{\mathsf{K}}
\newcommand{\Real}{\mathbb{R}}
\newcommand{\Rplus}{\Real_+}

\newcommand{\bC}{{\mathbb{C}}}
\newcommand{\bD}{{\mathbb{D}}}
\newcommand{\bK}{{\mathbb{K}}}
\newcommand{\bN}{{\mathbb{N}}}
\newcommand{\bQ}{{\mathbb{Q}}}
\newcommand{\bR}{{\mathbb{R}}}
\newcommand{\bT}{{\mathbb{T}}}
\newcommand{\bX}{{\mathbb{X}}}
\newcommand{\bZ}{{\mathbb{Z}}}
\newcommand{\bH}{{\mathbb{H}}}
\newcommand{\BH}{{\B(\H)}}
\newcommand{\bsl}{\setminus}
\newcommand{\ca}{\mathrm{C}^*}
\newcommand{\cstar}{\mathrm{C}^*}
\newcommand{\cenv}{\mathrm{C}^*_{\text{env}}}
\newcommand{\rip}{\rangle}
\newcommand{\ol}{\overline}
\newcommand{\td}{Widetilde}
\newcommand{\Wh}{Widehat}
\newcommand{\sot}{\textsc{sot}}
\newcommand{\Wot}{\textsc{wot}}
\newcommand{\Wotclos}[1]{\ol{#1}^{\textsc{wot}}}
 \newcommand{\A}{{\mathcal{A}}}
 \newcommand{\B}{{\mathcal{B}}}
 \newcommand{\C}{{\mathcal{C}}}
 \newcommand{\D}{{\mathcal{D}}}
 \newcommand{\E}{{\mathcal{E}}}
 \newcommand{\F}{{\mathcal{F}}}
 \newcommand{\G}{{\mathcal{G}}}
\renewcommand{\H}{{\mathcal{H}}}
 \newcommand{\I}{{\mathcal{I}}}
 \newcommand{\J}{{\mathcal{J}}}
 \newcommand{\K}{{\mathcal{K}}}
\renewcommand{\L}{{\mathcal{L}}}
 \newcommand{\M}{{\mathcal{M}}}
 \newcommand{\N}{{\mathcal{N}}}
\renewcommand{\O}{{\mathcal{O}}}
\renewcommand{\P}{{\mathcal{P}}}
 \newcommand{\Q}{{\mathcal{Q}}}
 \newcommand{\R}{{\mathcal{R}}}
\renewcommand{\S}{{\mathcal{S}}}
 \newcommand{\T}{{\mathcal{T}}}
 \newcommand{\U}{{\mathcal{U}}}
 \newcommand{\V}{{\mathcal{V}}}
 \newcommand{\W}{{\mathcal{W}}}
 \newcommand{\X}{{\mathcal{X}}}
 \newcommand{\Y}{{\mathcal{Y}}}
 \newcommand{\Z}{{\mathcal{Z}}}

\newcommand{\fA}{{\mathfrak{A}}}
\newcommand{\fB}{{\mathfrak{B}}}
\newcommand{\fC}{{\mathfrak{C}}}
\newcommand{\fD}{{\mathfrak{D}}}
\newcommand{\fE}{{\mathfrak{E}}}
\newcommand{\fF}{{\mathfrak{F}}}
\newcommand{\fG}{{\mathfrak{G}}}
\newcommand{\fH}{{\mathfrak{H}}}
\newcommand{\fI}{{\mathfrak{I}}}
\newcommand{\fJ}{{\mathfrak{J}}}
\newcommand{\fK}{{\mathfrak{K}}}
\newcommand{\fL}{{\mathfrak{L}}}
\newcommand{\fM}{{\mathfrak{M}}}
\newcommand{\fN}{{\mathfrak{N}}}
\newcommand{\fO}{{\mathfrak{O}}}
\newcommand{\fP}{{\mathfrak{P}}}
\newcommand{\fQ}{{\mathfrak{Q}}}
\newcommand{\fR}{{\mathfrak{R}}}
\newcommand{\fS}{{\mathfrak{S}}}
\newcommand{\fT}{{\mathfrak{T}}}
\newcommand{\fU}{{\mathfrak{U}}}
\newcommand{\fV}{{\mathfrak{V}}}
\newcommand{\fW}{{\mathfrak{W}}}
\newcommand{\fX}{{\mathfrak{X}}}
\newcommand{\fY}{{\mathfrak{Y}}}
\newcommand{\fZ}{{\mathfrak{Z}}}

\newcommand{\sgn}{\operatorname{sgn}}
\newcommand{\rank}{\operatorname{rank}}
\newcommand{\supp}{\operatorname{supp}}
\newcommand{\dist}{\operatorname{dist}}
\newcommand{\Aut}{\operatorname{Aut}}
\newcommand{\Aff}{\operatorname{Aff}}
\newcommand{\Cknet}{{\mathcal{C}_{{\rm Knet}}}}
\newcommand{\Ckag}{{\mathcal{C}_{{\rm kag}}}}
\newcommand{\GL}{\operatorname{GL}}
\newcommand{\spn}{\operatorname{span}}

\title[Crystal flex bases and the RUM spectrum]{Crystal flex bases and the RUM spectrum}

\author[G. Badri, D. Kitson and S. C. Power]{G. Badri, D. Kitson and S. C. Power}
\thanks{Supported by EPSRC grant EP/P01108X/1.}

\address{Dept. of Mathematical Sciences.\\Umm Al-Qura University \\
Saudi Arabia }

\email{gmbadri@uqu.edu.sa}

\address{Dept.\ Math.\ Stats.\\ Lancaster University\\
Lancaster LA1 4YF \\U.K. }
\email{d.kitson@lancaster.ac.uk}
\email{s.power@lancaster.ac.uk}

\date{}

\begin{abstract}
A theory of free spanning sets, free bases and their space group symmetric variants is developed for the first order flex spaces of infinite bar-joint frameworks.  
Such spanning sets and bases are computed for a range of fundamental crystallographic bar-joint frameworks, including the honeycomb (graphene) framework, the octahedron (perovskite) framework and the 2D and 3D kagome frameworks. It is also shown that the existence of crystal flex bases is closely related to linear structure in the rigid unit mode (RUM) spectrum and a more general \emph{geometric flex spectrum}. 
\end{abstract}

\maketitle

\section{Introduction}\label{s:intro}
 The analysis of the rigidity and flexibility of periodic infinite bond-node structures is an ongoing endeavour in materials science and pure mathematics
(Connelly, Ivi\'c  and Whiteley \cite{con-whi}, Guest, Fowler and Power \cite{gue-fow-pow}). In this connection
it is well known in crystallography and engineering that many crystals and periodic structures which are critically coordinated exhibit a rich set of localised zero energy crystal vibrations (phonons) (Giddy et-al \cite{gid-et-al}, Dove et-al \cite{dov-hei-ham}, \cite{dov-exotic}, Wegner \cite{weg}) and localised infinitesimal mechanisms (Hutchinson and Fleck \cite{hut-fle}).
Inspired by this we wish to understand for which crystal frameworks there exists a finite or countable set of localised modes $u_1, u_2,\dots $ which tells the whole story, so to speak, by providing a \emph{complete set} in the sense that every first order (ie. infinitesimal) flex may be expressed as an infinite linear combination
\[
u = \sum_{n=1}^\infty\alpha_nu_n,\quad \alpha_n \in \bR.
\]
In the dynamical theory of a crystal framework $\C$ it is common to limit attention to first order motions with some form of periodicity, or periodicity up to a phase factor for a Bloch wave vector. 
However, one of our motivations is to determine the  structure of the space $\F(\C;\bR)$ of \emph{all} first order flexes, including unbounded ones, as well as the structure of the space $\F_\infty(\C;\bR)$ of all bounded first order flexes.

To approach these problems we formalise, in Section \ref{s:freebases}, the notions of a \emph{free spanning set} and a \emph{free basis} for a space of velocity fields for an arbitrary countable bar-joint framework, with no periodicity assumptions, while in Section \ref{s:crystalbases} we define their space group symmetric counterparts in the case of crystal frameworks. Also, in Section \ref{s:examples} we determine such spanning sets and bases for a range of fundamental examples, including the honeycomb (graphene) framework, the octahedron (perovskite) framework and the 2D and 3D kagome frameworks. Our main result, Theorem \ref{thm:crystal}, shows that if the infinitesimal flex space of a crystallographic bar-joint framework is infinite dimensional then in any space group symmetric free spanning set there necessarily exists a \emph{band-limited} flex, that is, one whose support lies uniformly close to  a proper linear subspace. We show that this requirement implies that there can be obstacles to the existence of crystal flex bases which arise from nonlinearity in the RUM spectrum $\Omega(\C)$, or nonlinearity in a more general \emph{geometric flex spectrum} $\Gamma(\C)$. 
See Theorems \ref{t:omegalinear} and  \ref{t:gammalinear}. 

Recall that the rigid unit modes, or RUMs, of a material crystal are the zero energy oscillation modes observed in the long wavelength limit \cite{gid-et-al}, \cite{dov-exotic}, \cite{weg}.
These vibration modes are bounded and periodic modulo a multiphase factor and they correspond precisely to nonzero infinitesimal flexes with the corresponding boundedness and periodicity properties. See  \cite{bad-kit-pow}, \cite{owe-pow-crystal}, \cite{pow-poly}. The spectrum of multiphase factors for such modes, which is a subset of the $d$-torus known as the RUM spectrum, corresponds to the points of rank degeneracy of a matrix function $\Phi_\C(z)$ which is computable from a choice of periodicity and building block unit for $\C$.

The determination of the space of all real or complex infinitesimal flexes requires an understanding  of unbounded flexes and  in this connection we introduce the \emph{transfer function}
$\Psi_\C(z)$ associated with $\C$ and its periodic structure. This  is an analytic matrix-valued function on the $d$-fold product $\bC_*^d$ of the punctured complex plane $\bC_*=\bC\backslash \{0\}$ defined as the unique extension of  $\Phi_\C(z)$.
The points $z= \omega$ of rank degeneracy of the transfer function indicate the presence of nonzero \emph{geometrically periodic flexes}. 
Such periodicity is characterised by a set of equations of the form
\[
u(p_k) = \omega^k u(p_0) = \omega_1^{k_1}\cdots \omega_d^{k_d}u(p_0)
\]
which relate the velocity $u(p_0)$ of a node $p_0$ in the base unit cell to the velocity $u(p_k)$ of the corresponding node $p_k$ in the cell with label $k\in \bZ^d$. In particular if the multifactor $\omega=(\omega_1,\dots ,\omega_d)$ has $|\omega_i|>1$ for some $i$ then the local velocities $u(p_k)$  are unbounded (or zero) as $k_i\to +\infty$, with $k_j$ fixed for $j \neq i$, and are geometrically decaying as $k_i\to -\infty$. The geometrically periodic flexes are also referred to as \emph{factor periodic} flexes with multifactor $\omega$. The geometric spectrum $\Gamma(\C)$ referred to above is the set of such multifactors, that is, the set of points of rank degeneracy of the transfer function.

While unbounded flexes do not correspond to physical modes for the bulk crystal 
their consideration is nevertheless important for the identification of \emph{surface modes} associated with a boundary wall or free surface. Indeed, such modes can be identified with bounded restrictions of unbounded flexes of the bulk crystal. 
In particular, in 3 dimensions a point $\omega=(\omega_1, \omega_2, \omega_3)$ in $\Gamma(\C)$ with $|\omega_i|\neq 1$ for a single value of $i$ indicates the existence of a first order surface mode associated with a hyperplane boundary wall which is normal to the period vector of the bulk crystal corresponding to $i$. 
Thus, a further motivation for the introduction of the geometric flex spectrum and the identification of crystal flex spanning sets lies in their connections with surface modes and isolated Weyl modes 
(Rocklin et al \cite{roc-et-al}), and their analogies with boundary modes for topological insulators (Graf and Porter \cite{gra-por}, Lubensky et al \cite{lub-et-al}).

The development is organised as follows. In Section \ref{s:freebases} we show, using an abstract nonconstructive argument, that for any countable bar-joint framework, periodic or not, every infinite dimensional linear subspace  of the space of velocity fields has a free basis. In Section \ref{s:groupactions} we consider group actions on free spanning sets; proving several key results which are then applied in Section \ref{s:crystalbases} to crystal flex spanning sets for periodic frameworks. Our first main result, Theorem \ref{thm:crystal}, shows that a crystal flex spanning set for an infinite dimensional space of flexes necessarily contains localised velocity fields which moreover have geometric periodicity relative to their support. 
We also show that the existence of a natural crystal flex basis can lead to a simple description of the bounded flexes as those whose infinite expansion in the basis have bounded coefficients. Following this, the RUM spectrum and the {geometric flex spectrum}, associated with a periodic structure for $\C$, are introduced and we obtain necessary spectral conditions for the existence of various crystal spanning sets. We also pose here the intriguing problem of determining sufficient conditions, including spectral conditions, which ensure the existence of a crystal flex basis. In the final section, which is largely independent of earlier results, we compute crystal flex bases and spanning sets for several fundamental examples.

\section{Free spanning sets and free bases}
\label{s:freebases}
Let $A$ be a non-empty set and let $X$ be a finite dimensional vector space over a field $\bK$, where $\bK=\bR$ or $\bC$. Endow $X$ with a norm and the norm topology.
Let $X^A$ denote the topological vector space of maps $f:A\to X$ with the usual pointwise vector space operations and the topology of pointwise convergence. The support of a map $f$ is denoted $\supp(f)$.

\begin{defn}
A sequence $(f_n)$ in $X^A\backslash \{0\}$
\emph{tends to zero strictly} if,
for each $a\in A$, the sequence $(f_n(a))$ in $X$ has at most finitely many nonzero terms.
\end{defn}

\begin{lem}\label{l:elem}
Let $(f_n)$ be a sequence in $X^A\backslash \{0\}$. The following conditions are equivalent.

\begin{enumerate}[(i)]
\item $(f_n)$ tends to zero strictly.

\item $\sum_{n=1}^\infty \alpha_n f_n(a)$ converges in $X$ for every sequence of scalars $(\alpha_n)$ in $\bK$ and every $a\in A$.
\end{enumerate}
\end{lem}

\proof
If $(i)$ holds and $(\alpha_n)$ is an arbitrary sequence of scalars then, for each $a\in A$, the  series $\sum_{n=1}^\infty \alpha_n f_n(a)$ has finitely many non-zero terms and hence converges in $X$. This establishes $(ii)$. If $(i)$ does not hold, then there exists $a\in A$ and a subsequence $(f_{n_j})$ such that $f_{n_j}(a)$ is non-zero for each $j$. Choose a norm $\|\cdot\|$ on $X$ and for each $n\in\bN$ define, 
\[\alpha_{n} = \left\{\begin{array}{cl}
\frac{1}{\|f_{n_j}(a)\|} & \mbox { if } n=n_j \mbox{ for some } j\in\bN,\\
0 & \mbox{ otherwise. }
\end{array}\right.\]
Then the series $\sum_{n=1}^\infty \alpha_n f_n(a)$ fails to converge in $X$
and so $(ii)$ is not satisfied. 
\endproof

If a sequence $(f_n)$ in $X^A\backslash \{0\}$ tends to zero strictly then, by Lemma \ref{l:elem}, for every sequence of scalars $(\alpha_n)$ in $\bK$ the formal sum $\sum_{n=1}^\infty \alpha_n f_n$ represents an element of $X^A$.
Note that the partial sums, $s_N=\sum_{n=1}^N \alpha_n f_n$, converge to $s=\sum_{n=1}^\infty \alpha_n f_n$ in the strict sense that the sequence $(s-s_N)$ tends to zero strictly in $X^A$.
Also note that, setting $S=\{f_n:n\in\bN\}$, the set $\M(S)= \{\sum_{n=1}^\infty \alpha_n f_n: \alpha_n\in \bK\}$ of all formal sums 
is a vector subspace of $X^A$.

\begin{defn}\label{d:freespanning}
Let $W$ be an infinite dimensional vector subspace of $X^A$.

\begin{enumerate}[(a)]
\item A \textit{free spanning set} for $W$
is a countable subset $S=\{f_1, f_2, \dots \}$ of $W\backslash \{0\}$
for which the sequence $(f_n)$ tends to zero strictly and {satisfies $W\subseteq \M(S)$}.

\item A \textit{free basis} for $W$ is a free spanning set $S=\{f_1, f_2, \ldots \}$ for $W$ such that each $w\in W$ has a unique representation,
\[
w = \sum_{n=1}^\infty \alpha_n f_n,\quad \alpha_n\in \bK.
\]
\end{enumerate}
\end{defn}

\begin{eg}
Let $A=\{a_k:k\in\bN\}$ be a countable set and let $x_1,\ldots,x_d$ be a basis for $X$.
For each $n\in\bN$ and each $\sigma\in \{1,\ldots,d\}$ define,
\[e_{n,\sigma}(a_k) = \left\{
\begin{array}{ll}
x_\sigma & \mbox{ if } n=k, \\
0 & \mbox{ otherwise.}
\end{array}\right.\]
Then $\{e_{n,\sigma}:n\in\bN,\,\sigma\in \{1,\ldots,d\}\}$ is a free basis for $X^A$.
\end{eg}

\subsection{Existence of free bases}
Let $A$ be a countably infinite set and let $\alpha=(A_i)$ be an increasing sequence of finite subsets of $A$ which cover $A$ (i.e.~$A_i\subset A_{i+1}$ for all $i\in\bN$ and $A=\cup_{i\in\bN}\, A_i)$. For each $i\in\bN$, let $\pi_i:X^A\to X^{A_i}$ denote the natural restriction map. 
Let $W$ be an infinite dimensional vector subspace of $X^A$ and let $W_i=\pi_i(W)$ for each $i\in\bN$. 
For all $i,j\in\bN$ with $j\geq i$, let $\pi_{j,i}:W_j\to W_i$ be the natural restriction maps.
Note that the pair $((W_i),(\pi_{j,i}))$ is an {\em inverse sequence} of finite dimensional vector spaces.

\begin{lem}
\label{l:inverse} 
There exists a sequence $(B_j)$ of disjoint finite sets in $W$ with the following properties. 
\begin{enumerate}[(i)]
\item $\pi_i(B_j)=\{0\}$ for all $i<j$.
\item $\pi_j(B_1\cup\cdots\cup B_j)$ is a basis for $W_j$.
\item $|B_1\cup\cdots\cup B_j|=\dim W_j$ for each $j\in\bN$.
\end{enumerate}
\end{lem}

\proof 
Let $d_i=\dim W_i$ for each $i\in\bN$.
Choose a minimal finite subset $B_1 \subset W$ with the property that $\pi_1(B_1)$ is a  basis for $W_1$. 
Since the restriction map $\pi_{2,1}:W_2\to W_1$ is linear and surjective there exists a subspace $Y_2$ of $W_2$ with $W_2=Y_2\oplus\ker\pi_{2,1}$ and $\dim Y_2 = \dim W_1 = d_1$.
Note that $\pi_1(w) = \pi_{2,1}\circ\pi_2(w)$ for all $w\in W$ and so it follows that $\pi_2(B_1)$ is a linearly independent set and hence a basis for $Y_2$.
Now choose a minimal finite subset $B_2 \subset W$ with the property that $\pi_2(B_2)$ is a basis for $\ker\pi_{2,1}$. Note that $B_1$ and $B_2$ are disjoint.
Then $\pi_1(B_2) = \pi_{2,1}\circ\pi_2(B_2)=0$, $\pi_2(B_1\cup B_2)$ is a basis for $W_2$ and $|B_1\cup B_2|=d_2$. Repeating these arguments the sequence $(B_j)$ can now be constructed inductively.
\endproof

\begin{thm}
\label{t:existence} 
If $A$ is a countably infinite set then every infinite dimensional vector subspace of $X^A$ has a free basis.
\end{thm}

\begin{proof} 
Let $W$ be a subspace of $X^A$ and let  $\alpha=(A_i)$ be an increasing sequence of finite subsets which cover $A$. There exists a sequence $(B_j)$ of disjoint finite sets in $W$ with properties $(i)$-$(iii)$ as in the statement of Lemma \ref{l:inverse}. It remains only to show that the countable set $\cup_{j\in\bN}\, B_j = \{f_n:n\in\bN\}$ is a free basis for $W$. By property $(i)$, the sequence $(f_{n})$ tends to zero strictly.
Suppose $w\in W$. 
By properties $(i)$-$(iii)$, there exists a unique sequence $(h_j)$ in $W$ such that 
$h_j\in \spn(B_j)$ and $\pi_j(w)=\pi_j(h_1+\cdots+h_j)$ for each $j\in\bN$. 
Define $h = \sum_{j=1}^\infty h_j$. Then $w=h\in \M(f_n)$ and so $W\subseteq\M(f_n)$. 
\end{proof}

\begin{rem}
Note that the proof of Theorem \ref{t:existence} is nonconstructive in the sense that the resulting free basis depends both on a choice of covering sequence $\alpha=(A_i)$ and  a process of selection for the sequence $(B_j)$. 
\end{rem} 

Given $W$ and $\alpha$ as above, denote by $\varprojlim W_i$ the {\em inverse limit} of the inverse sequence $((W_i),(\pi_{j,i}))$. Thus $\varprojlim W_i$ is the vector space of all $(f_i)$ in the Cartesian product $\Pi_{i\in\bN} W_i$ with the property that 
$\pi_{j,i}(f_j)=f_i$ whenever $j\geq i$. 
In the following, we consider the linear map $\Phi_\alpha: X^A\to \Pi_{i\in\bN} W_i$, $f\mapsto (\pi_i(f))$.

\begin{lem}
\label{lem:closed}
The following statements are equivalent.
\begin{enumerate}[(i)]
\item $W$ is closed in $X^A$.
\item $\Phi_\alpha(W) = \varprojlim W_i$.
\end{enumerate}
\end{lem}

\proof
To show $(i)\Rightarrow (ii)$, suppose $W$ is closed in $X^A$. 
Let $(\pi_i(w_i))\in \varprojlim W_i$, where $w_i\in W$ for each $i\in \bN$, and define $w\in X^A$ by setting $w(a) = w_j(a)$ for all $a\in A_j$, and all $j\in\bN$. Note that the sequence $(w_j)$ converges pointwise to $w$. Thus $w\in W$ and $\Phi_\alpha(w)=(\pi_i(w_i))$.

To show $(ii)\Rightarrow (i)$, suppose $\Phi_\alpha(W) = \varprojlim W_i$ and let $(w_n)$ be a sequence in $W$ which converges pointwise to $h\in X^A$. For each $i\in\bN$, the sequence $(\pi_i(w_n))$ in $W_i$ converges pointwise to $\pi_i(h)\in X^{A_i}$. Since $X^{A_i}$ is finite dimensional, $W_i$ is closed and so $\pi_i(h)\in W_i$. Thus the sequence $(\pi_i(h))$ lies in the inverse limit $\varprojlim W_i$ and so there exists $w\in W$ such that $\Phi_\alpha(w) = (\pi_i(h))$. Note that $\pi_i(h) = \pi_i(w)$ for all $i\in\bN$ and so $h=w\in W$.
\endproof

\begin{thm}
\label{t:existence_closed} 
If $A$ is a countably infinite set then every infinite dimensional closed vector subspace $W$ of $X^A$ has a free basis $S=\{f_n:n\in\bN\}$ with $W=\M(S)$.
\end{thm}

\begin{proof} 
Let $W$ be a closed subspace of $X^A$ and let  $\alpha=(A_i)$ be an increasing sequence of finite subsets which cover $A$. Let $S=\{f_n:n\in\bN\}$ be the free basis for $W$ obtained from a sequence $(B_j)$ as in the proof of Theorem \ref{t:existence}.
Let $h\in \M(S)$. 
By properties $(i)$ and $(ii)$ of the sequence $(B_j)$, $\pi_i(h)\in W_i$ for each $i\in\bN$ and so $(\pi_i(h))\in \varprojlim W_i$.
Since $W$ is closed, by Lemma \ref{lem:closed} there exists $w\in W$ such that $\Phi_\alpha(w)=(\pi_i(h))$. Now $h=w\in W$ and so $\M(S)\subseteq W$. 
\end{proof}

Given $f\in X^A$ we write $\|f\|_\infty = \sup_{a\in A} \|f(a)\|$. A mapping $f\in X^A$ is {\em bounded} if $\|f\|_\infty<\infty$. Otherwise, $f$ is said to be {\em unbounded}. 

\begin{cor}\label{c:unboundedelements}
Let $A$ be a countably infinite set. If $W$ is  an infinite dimensional closed subspace of $X^A$ then $W$ contains a countable linearly independent set of unbounded elements.
\end{cor}

\begin{proof}
By Theorem \ref{t:existence_closed}, there exists a free basis $S$ for $W$ with the property that $W=\M(S)$. Since $S$ is countably infinite we may choose  a sequence $(h_k)$ in $S$ and a sequence $(a_k)$ in $A$ so that
\[
h_k(a_k) \neq 0\quad  \mbox{and}\quad  h_k(a_j)=0, \quad  \mbox{for}\quad  1 \leq j \leq k-1.
\]
Replacing  $h_k$ by a scalar multiple of $h_k$ we may assume for convenience that
$\|h_k(a_k)\|=1$ for all $k$.
Choose non-zero scalars $(\alpha_{1,n})$ successively such that 
$\alpha_{1,1}=1$ and $|\alpha_{1,n}|\geq n+\|\sum_{k=1}^{n-1} \alpha_{1,k} h_k(a_n)\|$ for each $n\geq 2$.
Let $g_1= \sum_{k=1}^\infty \alpha_{1,k}h_{k}\in W$ and note that for each $n\in\bN$,
\[
\|g_1(a_n)\| = \|\sum_{k=1}^n \alpha_{1,k} h_{k}(a_n)\|
\geq \|\alpha_{1,n} h_n(a_n)\| - \|\sum_{k=1}^{n-1} \alpha_{1,k} h_k(a_n)\|
\geq n. 
\]
Thus $g_1$ is unbounded. For $m\geq 2$, similarly define 
$
g_m = \sum_{k=m}^\infty \alpha_{m,k} h_{k}
$
with non-zero scalars $(\alpha_{m,k})$ chosen so that $g_m$ is unbounded.
To see that the set $\{g_m:m\in\bN\}$ is linearly independent note that if
$\sum_{m=1}^d\lambda_mg_{m}=0$ for some $\lambda_1,\ldots,\lambda_d\in\bK$ then $\lambda_1\alpha_{1,1}h_{1}(a_{1})=\sum_{m=1}^d\lambda_mg_{m}(a_{1}) = 0$ 
and so $\lambda_1=0$. By similar arguments it follows that $\lambda_j=0$ for each $j=2,\ldots,d$.
\end{proof}

\subsection{Application to bar-joint frameworks}
The results of the previous section apply to the infinitesimal flex spaces of countable bar-joint frameworks (and in particular, to crystallographic bar-joint frameworks). 
Let $G=(V,E)$ be a simple graph with  vertex set $V$ and edge set $E$. 
A {\em  bar-joint framework} for $G$ in $\bR^d$ is a pair $\G=(G,p)$ where $p: V\to \bR^d$, $v \mapsto p_v$, is an injective map. 
The points $p_v$ are referred to as {\em joints} of $\G$ and the line segments $(p_v,p_w)$, where $vw\in E$, are referred to as {\em bars}.
Let $\V(\G;\bK)=X^A$ where $A=p(V)$ and $X=\bK^d$ (for $\bK=\bR$ or $\bC$). 
The elements of $\V(\G;\bK)$ are referred to as \textit{velocity fields} for $\G$ over $\bK$. 
If $u\in \V(\G;\bK)$ then $u(p_v)$ is referred to as the \emph{velocity vector} for $u$ at $p_v$. 
An \textit{infinitesimal flex} (or \emph{first order flex}) for $\G$ over $\bK$ is a velocity field $u\in\V(\G;\bK)$ with the property that 
$\langle p_v-p_w, u(p_v)\rangle = \langle p_v-p_w, u(p_w)\rangle$ for each edge $vw\in E$. 
The set of infinitesimal flexes for $\G$ is a closed vector subspace of $\V(\G;\bK)$,  denoted $\F(\G;\bK)$.
The vector subspace of bounded infinitesimal flexes for $\G$ is denoted $\F_\infty(\G;\bK)$.

\begin{center}
\begin{figure}[ht]
\centering
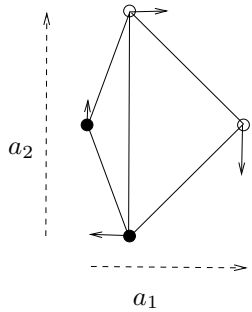
 \caption{A motif for the bar-joint framework  $\C_{\rm kite}$ together with four velocity vectors from the unbounded infinitesimal flex $a_{\rm kite}$.}
 \label{f:kitemotif}
\end{figure}
\end{center} 

\begin{eg}\label{eg:kite}
Consider the crystallographic bar-joint framework $\C_{\rm kite}$ in the Euclidean plane which is defined by the motif indicated in Figure \ref{f:kitemotif} consisting of two joints and five bars, together with the indicated pair $\{a_1, a_2\}$ of period vectors, where $\|a_1\| < \|a_2\|$.
Let $\vec{x}, \vec{y}$ be a choice of nonzero infinitesimal translation flexes of $\C_{\rm kite}$ for the $x$ and $y$ directions respectively and let $\vec{r}$ be a nonzero infinitesimal rotation flex.  Additionally, note that there is an unbounded infinitesimal flex $a_{\rm kite}$ in $\F(\C_{\rm kite};\bR)$, indicated in Figure \ref{f:kitemotif}, which restricts to alternating rotations of the $5$-bar motif and its translates. The magnitude of these rotations tend to infinity exponentially in the positive $x$-direction and are constant in magnitude in the $y$-direction. It is straightforward to show that the finite set $\{\vec{x}, \vec{y}, \vec{r}, a_{\rm kite}\}$ is a vector space basis for $\F(\C_{\rm kite};\bR)$. 
\end{eg}

Let $\G=(G,p)$ be a finite or countable bar-joint framework whose joints do not lie in a hyperplane. The space $\F_{\rm rig}(\G;\bR)$ of trivial infinitesimal flexes, or \emph{rigid motion infinitesimal flexes}, is the vector subspace of $\F(\G;\bR)$ consisting of  real infinitesimal flexes of the complete graph bar-joint framework $(K_{V},p)$. Here $K_{V}$ is the complete graph on the vertex set $V$. 
The bar-joint framework $\G$ is \emph{infinitesimally rigid} if $\F(\G;\bR) = \F_{\rm rig}(\G;\bR)$ and \emph{boundedly infinitesimally rigid} if $\F_\infty(\G;\bR) \subset \F_{\rm rig}(\G;\bR)$.

\begin{eg}\label{eg:surface}
The three velocity fields $\vec{x}, \vec{y}, \vec{r}$ in Example \ref{eg:kite} span $\F_{\rm rig}(\C_{\rm kite};\bR)$. Thus $\C_{\rm kite}$ is boundedly infinitesimally rigid, but not infinitesimally rigid. 
In contrast, the semi-crystallographic framework  $\C_{\rm kite}^{x\leq 0}$, formed by removing from $\C_{\rm kite}$ all the 5-bar kite frameworks not lying in the half-plane $x\leq 0$ is not boundedly infinitesimally rigid. Indeed, note that $\C_{\rm kite}^{x\leq 0}$ has a nontrivial bounded infinitesimal flex given by the restriction of $a_{\rm kite}$. We view this infinitesimal flex (and its associated zero energy mechanical mode) as a \emph{bounded surface flex} (or surface mode),  of the bulk crystal $\C_{\rm kite}$, which is associated with the domain wall $x=0$.
\end{eg}

\begin{thm}
Let $\G=(G,p)$ be a bar-joint framework in $\bR^d$ with a countable set of joints.
If the space of infinitesimal flexes $\F(\G;\bK)$ is infinite dimensional then,
\begin{enumerate}[(i)]
\item $\F(\G;\bK)$ has a free basis, and,
\item $\F(\G;\bK)$ contains a countable linearly independent set of unbounded infinitesimal flexes.
\end{enumerate}
\end{thm}

\proof
The  statements follow from
Theorem \ref{t:existence_closed} and Corollary \ref{c:unboundedelements} since $\F(\G,\bK)$ is  closed in $\V(\G;\bK)$.
\endproof

Several examples of free bases for the infinitesimal flex spaces of crystallographic bar-joint frameworks are presented in Section \ref{s:examples}. 

\section{Group actions on free spanning sets}\label{s:groupactions}
Once again let $A$ be a countably infinite set and $X$ a finite dimensional normed space over $\bK$. Throughout this section, $\Gamma$ denotes a multiplicative abelian group with identity element $1$, $\theta:\Gamma\times A\to A$ is a free group action on $A$
and $\pi:\Gamma\times X^A\to X^A$ is the induced faithful group action on $X^A$ given by
$\pi( \gamma, f )(a) = f(\gamma^{-1}a)$ for all $a\in A$. 
To simplify notation, $\theta(\gamma,a)$  will be written as $\gamma a$  for all $\gamma\in \Gamma$ and  $a\in A$. Similarly, $\pi(\gamma,f)$  will be written as $\gamma f$  for all $\gamma\in \Gamma$ and  $f\in X^A$. The {\em orbit} of an element $a\in A$ under the group action $\theta$ is the set $\Gamma a = \{\gamma a:\gamma\in \Gamma\}$ and the
{\em quotient set} $A/\Gamma=\{\Gamma a:a\in A\}$ is the set of all such orbits.

\begin{defn}
A {\em geometric direction} for $f\in X^A\backslash \{0\}$ is an element $\gamma\in \Gamma$ such that $\gamma f$ is a non-zero scalar multiple of $f$.
\end{defn}
The set of all geometric directions for $f$ is denoted $\Gamma_f$.

\begin{lem}
\label{l:GeomDirSub1}
Let $f\in X^A\backslash \{0\}$.
\begin{enumerate}[(i)]
\item $\Gamma_f$ is a subgroup of $\Gamma$.
\item $\Gamma_f= \Gamma_{\lambda(\gamma f)}$ for all non-zero scalars $\lambda\in \bK$ and all $\gamma\in \Gamma$.
\item If $f$ has finite support then $\Gamma_f=\{1\}$.
\end{enumerate}
\end{lem}

\proof
The proofs of parts $(i)$ and $(ii)$ are elementary and left to the reader.
To show $(iii)$, suppose $\gamma\in \Gamma$ is a geometric direction for $f$ and choose $a\in\supp(f)$. Then $\gamma^{k}a\in\supp(f)$ for each $k\in \bZ$. Since $\theta$ is a free group action, the elements $\gamma^{k}a$, $k\in \bZ$, are distinct unless $\gamma=1$.
\endproof

\begin{defn}
Let $S\subset X^A\backslash \{0\}$.
A group action $\pi:\Gamma\times X^A \to X^A$  {\em acts} on  $S$ (up to scalar multiples) if for each $f\in S$ and each $\gamma\in\Gamma$ there exists a scalar $\lambda\in\bK$ and $g\in S$ such that $\gamma f=\lambda g$. 
\end{defn}

\begin{lem}
\label{l:GeomDirSupp}
Let $S$ be a free spanning set for an infinite dimensional vector subspace of $X^A$ and let $f\in S$.
If the quotient set $A/\Gamma$ under the action $\theta$ is finite and the induced action $\pi$ acts on $S$ up to scalar multiples then the following statements are equivalent.
\begin{enumerate}[(i)]
\item $f$ has finite support.
\item $\Gamma_f=\{1\}$.
\end{enumerate}
\end{lem}

\proof
The implication $(i)\Rightarrow (ii)$ is Lemma \ref{l:GeomDirSub1}$(iii)$. To show $(ii)\Rightarrow (i)$, suppose $\Gamma_f=\{1\}$.
Let $S=\{f_n:n\in\bN\}$ and let $a_1,\ldots,a_{r}$ be a set of representatives for the finitely many orbits in $A/\Gamma$.
For $i\in\{1,\ldots,r\}$, define
$C_i=\{\gamma f:\gamma\in\Gamma, \, (\gamma f)(a_i)\not=0\}$.
Note that, since  the sequence $(f_n)$  tends to zero strictly, the set $S_i=\{g\in S:g(a_i)\not=0\}$ is finite and, since $\pi$ acts on  $S$, each element of $C_i$ is a scalar multiple of an element of $S_i$. 
If $C_i$ is not finite, for some $i\in\{1,\ldots,r\}$, then there exist distinct elements $\gamma_1 f,\gamma_2 f\in C_i$ and a non-zero scalar $\lambda\in\bK$ such that $\gamma_1 f=\lambda (\gamma_2 f)$. 
Thus $(\gamma_1\gamma_2^{-1}) f = \lambda f$ and so $\gamma_1\gamma_2^{-1}$ is a   geometric direction for $f$, which contradicts $\Gamma_f=\{1\}$.
Thus $C_i$ must be finite for each $i\in\{1,\ldots,r\}$. 
Now for each $C_i$, either $\gamma_1f = \gamma_2f$ for some distinct $\gamma_1,\gamma_2\in \Gamma$ or there are only finitely many $\gamma\in \Gamma$ with $(\gamma f)(a_i)\not=0$.
If the former case holds for some $C_i$, then $f = (\gamma_1\gamma_2^{-1}) f$ and so $\gamma_1\gamma_2^{-1}$ is a geometric direction for $f$, which contradicts $\Gamma_f=\{1\}$. 
Thus the latter case holds for each $C_i$. Since $f(\gamma^{-1} a_i)= (\gamma f)(a_i)$ for all $\gamma\in \Gamma$ and $i\in\{1,\ldots,r\}$, and 
since $A=\{\gamma a_i:\gamma\in \Gamma,\,i\in\{1,\ldots,r\}\}$, it follows that $f$ has finite support. 
\endproof

\begin{defn}
A set $S\subset X^A\backslash \{0\}$ is {\em finitely generated} by a group action $\pi:\Gamma\times X^A \to X^A$ (up to scalar multiples) if $\pi$ acts on $S$, and, there is a finite subset $S_0\subset S$ such that for every $f\in S$ there exists $\gamma\in \Gamma$, $f_0\in S_0$, and a scalar $\lambda\in \bK$, with $f= \lambda(\gamma f_0)$.
\end{defn}

\begin{lem}
\label{GeomDirSub2}
Let $S$ be a free spanning set for an infinite dimensional vector subspace of $X^A$. If $S$ is finitely generated by a group action $\pi$ then there exists $f\in S$ such that the quotient group $\Gamma/\Gamma_{f}$ has infinite order. 
\end{lem}

\proof
Let $\I = \{\bK f: f\in S\}$ where $\bK f$ denotes the one-dimensional subspace of $X^A$ spanned by $f$. If $\I$ is a finite set, then there exists a sequence $(f_{n})$ in $S$ for which the one-dimensional spaces $\bK {f_{n}}$ are equal.
In particular, the support sets $\supp(f_{n})$ are equal for all $n\in\bN$.
This is a contradiction since $(f_n)$ tends to zero strictly and so $\I$ is countably infinite. 
Since $S$ is finitely generated by $\pi$, there exists a finite subset $S_0\subset S$ such that $\I = \{\bK(\gamma f): \gamma\in\Gamma, \, f\in S_0\}$.
Note that, for each $f_0\in S_0$, the cardinality of $\{\bK(\gamma f_0):\gamma\in\Gamma\}$ is bounded by the order of $\Gamma/\Gamma_{f_0}$. Thus $\Gamma/\Gamma_{f_0}$ has infinite order for some $f_0\in S_0$.
\endproof

\subsection{Lattice group actions}
Let $L=\{\sum_{i=1}^m n_i b_i:n_i\in\bZ\}$ be a rank $m$ lattice in $\bR^d$ determined by linearly independent vectors $b_1,\ldots,b_m\in\bR^d$.
For each $b\in L$, let $T_b:\bR^d\to\bR^d$ denote the translation $T_b(x)=x+b$. 
The translation group associated to $L$ is the multiplicative abelian group $\T=\{T_b:b\in L\}$. 
A \emph{multilattice} on $L$ is a set $A = A_1 \cup \ldots \cup A_s$ where $A_1,\ldots, A_s$ are pairwise disjoint translates of $L$.
Define a free group action $\theta_L:\T\times A\to A$ by setting $\theta_L(T,a_i) = T(a_i)$ for $a_i \in A_i$, $1 \leq i \leq s$.
The {\em lattice group action} on $X^A$ induced by $\theta_L$ is the group action $\pi_L:\T\times X^A\to X^A$ with $\pi_L(T,f) = f\circ T^{-1}$. Note that the quotient set $A/\T$ is finite and a translation $T_b\in \T$ is a geometric direction for $f\in X^A$ if and only if there exists a nonzero scalar $\lambda$ such that $f(a-b)=\lambda f(a)$ for all $a\in A$.  If $\Gamma_f\not=\{1\}$ then the sublattice $L_f=\{b\in L:T_b\in \Gamma_f\}$ is referred to as the {\em lattice of geometric directions} for $f$.

\begin{defn}
Let $L'$ be a rank $n$ sublattice of $L$ with basis $b_1,\ldots,b_n\in\bR^d$. A mapping $f\in X^A$ is {\em factor periodic} on $L'$ if there exists $\omega=(\omega_1,\ldots,\omega_n)\in\bK_*^n$ such that 
$T_bf = \omega_1^{k_1}\cdots\omega_n^{k_n} f$ for all $b=\sum_{j=1}^n k_jb_j\in L'$. In this case, $\omega$ is referred to as the {\em periodicity multifactor} for $f$ determined by the basis $b_1,\ldots,b_n\in\bR^d$.
\end{defn}

\begin{lem}
\label{l:factorperiodic}
Let $A$ be a multilattice on $L$ and let $f\in X^A\backslash\{0\}$.
\begin{enumerate}[(i)]
\item
If $\Gamma_f\not=\{1\}$ then $f$ is factor periodic on $L_f$.
\item
If $f$ is bounded and factor periodic, with periodicity multifactor $\omega$, then the components of $\omega$ are unimodular.
\end{enumerate}
\end{lem}

\proof
$(i)$ If  $\Gamma_f\not=\{1\}$ then $L_f$ is a lattice of rank $n$ say.
Let $b_1,\ldots,b_n\in \bR^d$ be a set of generators for the lattice $L_f$.
For each $i=1,\ldots,n$, $T_{b_i}$ is a geometric direction for $f$ and so $T_{b_i}f=\omega_i f$ for some non-zero scalar $\omega_i\in\bK$. 
Thus $T_bf = T_{b_1}^{k_1}\cdots T_{b_n}^{k_n}f
=\omega_1^{k_1}\cdots\omega_n^{k_n} f$ for all $b=\sum_{j=1}^n k_jb_j\in L_f$.

$(ii)$ Suppose $f$ is factor periodic on the lattice $L'$ and let $\omega=(\omega_1,\ldots,\omega_n)\in\bK_*^n$ be the periodicity factor for $f$ determined by a basis $b_1,\ldots,b_n$ for $L'$. Let $a\in \supp(f)$. Note that, for each $j=1,\ldots,n$, $\|f(a-kb_j)\|
=\|T^k_{b_j}(a)\|=|\omega_j|^k\|f(a)\|$ for all $k\in\bZ$. Thus, if $|\omega_j|\not=1$ for some $j$ then $f$ is clearly unbounded.

\endproof

If $S\subset X^A$ then, for each $a\in A$, let $S_a=\{f\in S:f(a)\not=0\}$. 

\begin{lem}
\label{l:supportmult}
Let $A$ be a multilattice on $L$ and let $S$ be a free spanning set for an infinite dimensional subspace of $X^A$.
If the lattice group action $\pi_L$ acts on $S$ then each $f\in S$ satisfies one of the following conditions.
\begin{enumerate}[(i)]
\item $f$ is finitely supported.
\item $f$ is factor periodic on $L_f$ and the support of $f$ is a multilattice on $L_f$.
\end{enumerate}
\end{lem}

\proof
Suppose $f\in S$ is not finitely supported. 
By Lemma \ref{l:GeomDirSupp}, $\Gamma_f\not=\{1\}$. Thus, by Lemma \ref{l:factorperiodic}, $f$ is factor periodic on $L_f$.
Choose $a_1,\ldots,a_s\in A$ such that $A=\cup_{i=1}^s (a_i+L)$. 
Since $S$ is a free spanning set, $S_{a_i}$ is finite and so $S_{a_i}=\{g_{a_i,1},\ldots,g_{a_i,s_i}\}$ say.
By Lemma \ref{l:GeomDirSupp}, the subgroup $L_f$ is a non-zero sublattice of $L$.
Let $L_{a_i,k}=\{b\in L: T_b^{-1}f = \lambda g_{a_i,k}, \mbox{ for some }\lambda\not=0\}$ for $k=1,\ldots,s_i$.
Note that $(T^{-1}_b f)(a_i)=f(a)\not=0$ for each $a\in\supp(f)$ with $a=a_i+b$ and $b\in L$.
Thus, since $\pi_L$ acts on $S$, it follows that $\supp(f) = \cup_{i=1}^s \cup_{k=1}^{s_i} (a_i+L_{a_i,k})$.
If $b,b'\in L_{a_i,k}$ then $T_bf = \lambda (T_{b'} f)$ for some non-zero scalar $\lambda\in\bK$ and so $T_{b-b'}$ is a geometric direction for $f$. In particular, $b-b'\in L_f$. Thus if $b_{a_i,k}\in L_{a_i,k}$ then it follows that $L_{a_i,k}=b_{a_i,k}+L_f$, and so $\supp(f)$ is a finite union of translates of the lattice $L_f$.
\endproof

In the following, $\bR^d$ is endowed with a norm and $\dist(a,E)= \inf_{x\in E} \|a-x\|$ denotes the distance between a point $a\in \bR^d$ and a subset $E\subset \bR^d$.

\begin{defn} Let $A$ be a multilattice in $\bR^d$.
A mapping $f\in X^A$ is {\em band-limited} if there exists a subspace $\K$ in $\bR^d$ and $C>0$ such that the support of $f$ is contained in the ``band" $\{x\in \bR^d: \dist(x,\K)\leq C\}$. 
\end{defn}

\begin{thm}
\label{thm:multilattice}
Let $A$ be a multilattice and let $S$ be a free spanning set for an infinite dimensional vector subspace of $X^A$. 
If $S$ is  finitely generated by the lattice group action $\pi_L$ then $S$ contains a band-limited mapping.
\end{thm}

\proof
Note that if a map $f\in S$ has finite support then it is band-limited. So suppose no element of $S$ has finite support. 

By Lemma \ref{GeomDirSub2}, there exists $f\in S$ such that the  quotient group $\Gamma/\Gamma_{f}$ has infinite order.
Recall that a quotient of two free abelian groups with equal rank has finite order. Thus $\Gamma_{f}$ has rank $k$ strictly less than $d$. 
Let $\K$ be the proper linear subspace in $\bR^d$ spanned by $L_{f}$.
By Lemma \ref{l:supportmult}, the support of $f$ is a multilattice on $L_{f}$ and hence is contained in a band $\{x\in \bR^d:d(x,\K)\leq C\}$ for some $C>0$.
\endproof

A subset $S$ of $X^A$ is {\em bounded} if each element of $S$ is bounded and $\sup_{f\in S}\|f\|_\infty<\infty$.

 \begin{lem}\label{l:boundedflex}
Let $A$ be a multilattice on $L$ and let $S=\{f_n:n\in\bN\}$ be a bounded free spanning set for an infinite dimensional subspace of $X^A$.
If the lattice group action $\pi_L$ acts on $S$ and $(\alpha_n)\in \ell^\infty(\bN)$ is a bounded sequence of scalars then $\sum_{n=1}^\infty \alpha_nf_n$ is bounded.
\end{lem}

 \begin{proof}
Suppose $f = \sum_{n=1}^\infty \alpha_nf_n$ for some $\alpha=(\alpha_n)\in \ell^\infty(\bN)$.
Since $(f_n)$ tends to zero strictly, $S_{a}$ is a finite set for each $a\in A$.
Note that for each $b\in L$, $f(a+b)\not=0$ if and only if $(T_b^{-1}f)(a)\not=0$.
Thus, since $\pi_L$ acts on $S$, the set $S_{a+b}$ has the same cardinality as $S_{a}$.
Choose $a_1,\ldots,a_s\in A$ such that $A=\cup_{i=1}^s (a_i+L)$. 
Let $N=\sup_{a\in A} |S_{a}| = \max_{i=1,\ldots,s} |S_{a_i}|$
and let $M=\sup_{n}\|f_n\|_\infty$.
Then $\|f\|_\infty = \sup_{a\in A} \|f(a)\|
\leq NM\|\alpha\|_\infty$ and so $f$ is bounded. 
\end{proof}

In the following, a set $S\subset X^A\backslash \{0\}$ has the 
\emph{local basis property} if, for each $a\in A$, the set $\{f(a):f\in S_a\}$ is a basis for $X$. Also, $W_\infty$ denotes the set of bounded elements in a subspace $W$ of $X_A$.

 \begin{thm}\label{t:boundedflex}
Let $A$ be a multilattice on $L$ and let $S=\{f_n:n\in\bN\}$ be a free spanning set for an infinite dimensional subspace $W$ of $X^A$ with the following properties.
\begin{enumerate}[(i)]
\item The lattice group action $\pi_L$ acts on $S$.
\item $S$ is bounded and has the local basis property.
\item No element of $S$ is finitely supported.
\end{enumerate}
Then, 
 \[
 W_\infty =\{f\in W: f = \sum_{n=1}^\infty \alpha_nf_n \mbox{ and } (\alpha_n)\in \ell^\infty(\bN)\}.
 \]
\end{thm}

 \begin{proof}
Let $f \in W_\infty$. Then $f = \sum_{n=1}^\infty \alpha_nf_n$ for some sequence of scalars $(\alpha_n)$ and $\|f\|_\infty<\infty$.
Since $S$ has the local basis property, given any $a\in A$, the set $\{g(a):g\in S_a\}$ is a basis for $X$ and so we may consider the norm on $X$ given by $\|x\|_a = \max_{g\in S_a} |\lambda_g|$ where $x = \sum_{g\in S_a} \lambda_g g(a)$. 
Moreover, since $X$ is finite dimensional, there exists $c_a>0$ such that 
$c_a\|x\|_a\leq \|x\|$ for all $x\in X$.
For a given $m$ with $\alpha_m\neq 0$ there exists $a\in A$ with $f_m(a)\neq 0$.
Since the lattice group action $\pi_L$ acts on $S$ and no element of $S$ is finitely supported, it follows from Lemma \ref{l:GeomDirSupp} and Lemma \ref{l:factorperiodic}$(i)$ that each $g\in S_a$ is factor periodic on $L_g$.
Moreover, since each element of $S$ is bounded, it follows from Lemma \ref{l:factorperiodic}$(ii)$ that the components of the periodicity factors for each $g\in S_a$ are unimodular.
Choose $a_1,\ldots,a_s\in A$ such that $A=\cup_{i=1}^s (a_i+L)$. Then $a=a_j+b$ for some $j$ and some $b\in L$. 
Thus, since the lattice group action $\pi_L$ acts on $S$, there exists a unimodular scalar $\mu_j$ such that $S_a=\{\mu_j g:g\in S_{a_j}\}$. 
It follows that $\|x\|_a = \|x\|_{a_i}$ and so $c_{a_j}\|x\|_a=c_{a_j}\|x\|_{a_j}\leq \|x\|$ for each $x\in X$.
Thus $|\alpha_m|\leq \|f(a)\|_a \leq \frac{1}{c}\|f(a)\|\leq \frac{1}{c}\|f\|_\infty$
where $c=\min_{i=1,\ldots,s} c_{a_i}$. Thus $(\alpha_n)\in \ell^\infty(\bN)$.
 
For the reverse inclusion apply Lemma \ref{l:boundedflex}.
\end{proof}

\section{Crystal flex bases and the RUM spectrum}\label{s:crystalbases}
An {\em automorphism} of a graph $G=(V,E)$ is a bijection $\beta:V\to V$ with the property that $vw\in E$ if and only if $\beta(v)\beta(w)\in E$. 
The \emph{space group} for a bar-joint framework $\G=(G,p)$ in $\bR^d$ is the group $\fS(\G)$ of Euclidean isometries $T:\bR^d\to\bR^d$  with the property that $T(p(V))=p(V)$ and the induced map $\beta:V\to V$, 
$v\mapsto p^{-1}(T(p_v))$ is an automorphism of $G$. 
The subgroup of $\fS(\G)$ consisting of translations is denoted $\T(\G)$.

\begin{defn}
A  bar-joint framework $\C=(G,p)$ in $\bR^d$ is referred to as a \emph{crystal framework} if there exists a rank $d$ lattice $L$ in $\bR^d$ with the following properties.
\begin{enumerate}[(a)]
\item The translation group $\T=\{T_b:b\in L\}$ is a subgroup of $\T(\C)$.
\item $\C$ has only finitely many distinct vertex orbits and edge orbits under $\T$. \end{enumerate}
In this case, the lattice $L$ (or equivalently, the translation group $\T$) is referred to as a {\em periodic structure} for $\C$.
\end{defn}

Let $\C$ be a crystal framework in $\bR^d$. The {\em space group action} on $\V(\C;\bR)$ is the faithful group action $\pi_{\fS(\C)}: \fS(\C) \times \V(\C;\bR) \to \V(\C;\bR)
$  with $\pi_{\fS(\C)}(T, f)  = T_2\circ f\circ T^{-1}$  where $T=T_1T_2$ is the unique factorisation of the Euclidean isometry $T$ with $T_1$ a translation and $T_2$ an orthogonal linear transformation. 
Also we define the {\em space group action} on $\V(\C;\bC)$ as the natural induced action. Alternatively, for $f \in \V(\C;\bC)$ we may define $\pi_{\fS(\C)}(T, f)  = \tilde{T}_2\circ f\circ T^{-1}$ where $\tilde{T}_2$ has the diagonal action $T_2 \oplus T_2$ on the direct sum $\bC^d = \bR^d\oplus i\bR^d$.
The {\em lattice group action} on $\V(\C;\bK)$ determined by a periodic structure $L$ with translation group $\T=\{T_b:b\in L\}$ is the group action $\pi_{\T}: \T \times \V(\C;\bK) \to \V(\C;\bK)$, $\pi_{\T}(T, f)=f\circ T^{-1}$.

\begin{lem}\label{lem:spacegroup}
Let $\C=(G,p)$ be a crystal framework in $\bR^d$ and let $S\subset\V(\C;\bK)$. If $S$ is finitely generated by the space group action $\pi_{\fS(\C)}$ then 
  $S$ is finitely generated by the lattice group action $\pi_{\T}$ for any choice of periodic structure $\T$.
\end{lem}

\begin{proof}
Suppose $S$ is finitely generated by $\pi_{\fS(\C)}$ and let $\T$ be a periodic structure for $\C$. Then given any $f\in S$ and any $T\in \fS(\C)$ there exists $\lambda\in \bK$ and $g\in S$ such that $Tf=\lambda g$. In particular, this holds for all $T\in \T$ and so $\pi_{\T}$ acts on $S$.
Also, there exists a finite subset $S_0\subset S$ such that every member of $S$ can be expressed as a scalar multiple of $Tf$ for some $T\in\fS(\C)$ and some $f\in S_0$.
Note that the set of joints of $\C$ is a multilattice and the quotient group $\fS(\C)/\T$ acts faithfully on the finite set of translation orbits. It follows that the quotient group $\fS(\C)/\T$ is finite. 
Let $R_1,\ldots,R_m$ be a set of representatives for the finitely many elements in $\fS(\C)/\T$  and let $S_0'=\{R_jf: f\in S_0, j=1,\ldots,m\}$.
Then every member of $S$ is a scalar multiple of $Tf$ for some $T\in\T$ and some $f\in S_0'$ and so $S$ is finitely generated over $\T$.
\end{proof}

A \emph{geometric flex} for a crystal framework $\C$  is an infinitesimal flex which is factor periodic with respect to sublattice of a periodic structure for $\C$. 
A \emph{local geometric flex} for $\C$ is an infinitesimal flex which is factor periodic with respect to a lower rank sublattice of a periodic structure for $\C$. 
A \emph{band-limited flex} for $\C$ is an infinitesimal flex which is also a band-limited vector field in $\V(\C;\bK)$.
For crystal frameworks it is of interest to determine free bases for $\F(\C,\bK)$ which incorporate localised or band-limited flexes, should such flexes exist, and which, moreover, incorporate the crystallographic symmetry group. Accordingly we make the following definitions.

\begin{defn}\label{d:crystalflexbasis}
Let $\C$ be a crystal framework and let $W$ be an infinite dimensional vector subspace of $\V(\C;\bK)$.

\begin{enumerate}[(a)]
\item A \emph{crystal spanning set} (respectively, \emph{crystal basis}) for $W$ is a free spanning set (respectively, free basis) which is finitely generated by the space group action $\pi_{\fS(\C)}$.

\item   A crystal spanning set (resp. crystal basis) for $W= \F(\C;\bK)$ is also referred to as a  \emph{crystal flex spanning set} (resp. \emph{crystal flex basis}) for $\C$.
\end{enumerate}
\end{defn}

In the next section we identify a number of crystal frameworks, such as the grid frameworks and the kagome frameworks which possess a crystal flex basis in the sense above. In several cases these bases consist entirely of local geometric flexes. Also we see that the flex space of the octahedron framework
has a free basis which is an \emph{essential crystal flex basis} in the sense that there is a subset which is a crystal flex basis for a vector subspace of $\F(\C;\bR)$ which is complementary to the $3$-dimensional space of infinitesimal rotation flexes.

\begin{thm}
\label{thm:crystal}
Let $\C=(G,p)$ be a crystal framework in $\bR^d$ and suppose $\F(\C;\bK)$ is infinite dimensional. Let $S$ be a crystal flex spanning set for $\C$.
\begin{enumerate}[(i)]
\item If $f\in S$ then $f$ is either finitely supported or a geometric flex for $\C$. 
\item $S$ must contain a band-limited flex. Moreover, $f\in S$ is a band-limited flex if and only if it is either finitely supported or a local geometric flex.

\item If $S$ is bounded, has the local basis property, and the elements of $S$ are not finitely supported then,
 \[
 \F_\infty(\C;\bK) = \{f\in \F(\C;\bK): f = \sum_{n=1}^\infty \alpha_nf_n \mbox{ and } (\alpha_n)\in \ell^\infty(\bN)\}.
 \]
\end{enumerate}
\end{thm}

\proof
By Lemma \ref{lem:spacegroup}, $S$ is finitely generated by the lattice group action $\pi_\T$ for any choice of periodic structure $\T$. 
The set of joints of $\C$ is a multilattice in $\bR^d$ and so the results of Section \ref{s:groupactions} may be applied with $A=p(V)$ and $X=\bK^d$. Statement $(i)$ follows from Lemma  \ref{l:GeomDirSupp}, Lemma \ref{l:factorperiodic} and Lemma \ref{l:supportmult}.
Statement $(ii)$ follows from $(i)$, Lemma \ref{l:supportmult} and Theorem \ref{thm:multilattice}. Statement $(iii)$ follows from Theorem \ref{t:boundedflex}.
\endproof

We now define a  \emph{transfer function} $\Psi_\C(z)$ associated with a crystal framework $\C$, a choice of periodic structure $L$ with translation group $\T$ and a choice of basis $b_1,\ldots,b_d$ for the lattice $L$. 
Let $F_v =\{p_\kappa: 1 \leq \kappa \leq |F_v|\}$ be a finite set
of joints representing the $\T$-translation classes of the joints of $\C$. Then we may conveniently
label the joints of $\C$ as,
\[
\{p_{\kappa, k}: 1 \leq \kappa \leq |F_v|, k\in \bZ^d\},
\]
where $p_{\kappa, k}= T_kp_\kappa$. Also let $F_e$ be a finite set  of representative bars for the translation classes of the  bars. The pair $(F_v, F_e)$, referred to as a \emph{motif} for $\C$, together with the periodic structure $\T$ carry  the essential geometric information which defines $\C$ \cite{owe-pow-crystal, pow-poly}.

\begin{defn}\label{d:transferfunction}
Let $\C$ be a crystal framework in $\bR^d$ with  motif 
$(F_v, F_e)$ and for $e=vw\in F_e$ let $p(e)=p(v)-p(w)$. The {\em transfer function} $\Psi_\C(z)$ is a matrix-valued function on $\bC_*^d$ whose rows are labelled by the edges of $F_e$ and whose columns are labelled by the vertex-coordinate pairs in $F_v\times \{1,\dots ,d\}$.
The row for an edge $e= (v,k)(w,l)$ with $v\neq w$ takes the form,
\[\kbordermatrix{& & & & v & & & & w & & & \\
e & 0 & \cdots &0 & p(e){z}^{-k} &0& \cdots&0 &- p(e){z}^{-l} &0& \cdots &0 },\]
while if $v=w$ it takes the form,
\[\kbordermatrix{& & & & v & & & \\
e & 0 & \cdots &0 & p(e)({z}^{-k} - {z}^{-l}) &0& \cdots&0 }.
\]
\end{defn}
\medskip

The restriction of $\Psi_\C(z)$ to the $d$-torus $\bT^d$ gives the {\em symbol function} $\Phi_\C(z)$ considered in \cite{bad-kit-pow,owe-pow-crystal,pow-poly}.
With the labelling of the joints of $\C$ given above, note that  a complex velocity field $u\in\V(\C,\bC)$ is a map $u:F_v\times \bZ^d \to \bC^d$ where $u(p_\kappa,k)$ is the velocity vector assigned to the joint $p_{\kappa,k}$. Note that $u$ is factor periodic, 
 with periodicity multifactor $\omega \in \bC_*^d$, if $u(p_\kappa,k)=\omega^k u(p_\kappa,0)$ for all  $p_\kappa \in F_v$, $k\in \bZ^d$. Here  $\omega^k$ is the product $\omega_1^{k_1} \dots \omega_d^{k_d}$.
We a write $u = b \otimes e_\omega$ for this vector field, where $b$ is the vector $(u(p_\kappa,0))_{\kappa\in F_v}$ in $\bC^{d|F_v|}$ and $e_\omega$ is the multi-sequence $(\omega^k)_{k\in\bZ^d}$. 

\begin{thm}
Let $\C$ be a crystal framework in $\bR^d$ with motif $(F_v,F_e)$ and let  $u=b\otimes e_\omega$ where $\omega\in\bC_*^d$ and $b\in \bC^{d|F_v|}$. The  following conditions are equivalent.
\begin{enumerate}[(i)]
\item $u\in\F(\C,\bC)$.
\item $\Psi({\omega}^{-1})b=0.$
\end{enumerate}
\end{thm}

\begin{proof}
The proof is similar to the unimodular case given in Power \cite{pow-poly}.
\end{proof}

\begin{defn}
Let $\C$ be a crystal framework in $\bR^d$ with a transfer function $\Psi_\C(z)$.
\begin{enumerate}[(a)]
\item 
The \emph{geometric flex  spectrum} $\Gamma(\C)$ of $\C$ is defined to be the set, 
 \[
 \Gamma(\C):=\{\omega \in \bC^d_* : \ker \Psi_\C({\omega}^{-1}) \neq \{0\}\}.
 \]
\item
The \emph{rigid unit mode spectrum}, or RUM spectrum, of $\C$ is  the set,
 \[
 \Omega(\C):=\{\omega \in \bT^d : \ker \Phi_\C(\overline{\omega}) \neq \{0\}\},
 \]
 where $\Phi_\C(z)$ is the restriction of $\Psi_\C(z)$ to the $d$-torus. 
\end{enumerate}
\end{defn}

Note that $\Gamma(\C)$ is  the set of points $\omega\in\bC_*^d$ for which there exists a nonzero factor periodic infinitesimal flex with periodicity multifactor $\omega$. This set depends both on the choice of periodic structure $L$ for $\C$ and the choice of basis $b_1,\ldots,b_d$ for $L$.  
The RUM spectrum is the subset  $\Omega(\C)=\Gamma(\C)\cap \bT^d$  of  multifactors $\omega$ with unimodular coordinates. Such multifactors are also referred to as \emph{multi-phases}, or simply \emph{phases}. 
Note also that for a critically coordinated crystal framework, in the elementary sense that $|F_e|=d|F_v|$, the transfer function is a square matrix-valued analytic function on its domain. It follows in this case that the determinant provides a multivariable analytic function and that the geometric spectrum is given by its set of zeros.

A small gallery of crystal frameworks and their transfer functions is given in Badri, Kitson and Power \cite{bad-kit-pow} and their RUM spectra are determined. See also in Power \cite{pow-poly} where the connection with rigid unit mode vibrations in materials science is discussed.
In Section \ref{s:examples} we examine the geometric 
flex spectrum and the RUM spectrum of a novel bipyramid framework. In close analogy with the semi-crystallographic kite framework of Example \ref{eg:surface} we see that there are geometrically decaying surface flexes associated with directions with no half-turn symmetry.

 \begin{defn}
Let $\C$ be a crystal framework in $\bR^d$ with a transfer function $\Psi_\C(z)$. The RUM spectrum $\Omega(\C)$ is said to \emph{contain linear structure} if the logarithmic representation of $\Omega(\C)$ in $[0,2\pi)^d$ contains a $t$-dimensional set of the form $[0,2\pi)^d \cap H$ where $H$ is an affine subspace of $\bR^d$ and $1\leq t \leq d$.
  
 \end{defn}

If the RUM spectrum for a particular periodic structure contains linear structure then the same is true for the RUM spectrum for any periodic structure. This follows from the fact that the RUM spectrum for a periodic structure arises as the image of the primitive RUM spectrum under a natural surjective map  \cite{pow-poly}. 

 \begin{thm}
 \label{t:omegalinear}
 Let $\C$ be a crystal framework in $\bR^d$ and let $S$ be a crystal flex spanning set for $\C$. If $S$ contains a bounded band-limited flex then $\Omega(\C)$ contains linear structure.
 \end{thm}

 \begin{proof}
Let $u\in S$ be a bounded band-limited flex for $\C$. If $u$ is finitely supported then, by \cite[Theorem 5.6]{pow-poly}, $\Omega(\C)=\bT^d$ and so $\Omega(\C)$ contains linear structure. Suppose $u$ is not finitely supported. By Theorem \ref{thm:crystal}$(ii)$, either $u$ is finitely supported or there is a periodic structure $\T$ for $\C$ such that $u$ is factor periodic with respect to a rank $t$ sublattice, where $1\leq t\leq d-1$. Let $\T_u$ be the subgroup of the translation group $\T$ determined by this sublattice. Also, let $\omega_0= (\omega_1, \dots ,\omega_t)$ be the periodicity multifactor. By Lemma \ref{l:factorperiodic}, since $u$ is bounded, $\omega_0\in\bT^t$. Assume that $\T_u$ has generators $T_{g(1)}, \ldots , T_{g(t)}$ and choose translations $T_{g(t+1)}\dots ,T_{g(d)}$ such that $T_{g(1)}, \ldots , T_{g(d)}$ is a set of generators for a full rank subgroup $\T'$ of $\T$. 
Let $\omega_* = (\omega_{t+1},\ldots , \omega_d)$ be an arbitrary point in $\bT^{d-t}$ and define
\[
w = \sum_{k'\in \bZ^{d-t}} \omega_*^{k'}X_{k'}u,
\] 
where $X_{k'}= k'_1T_{g(t+1)}+\ldots + k'_{d-t}T_{g(d)}$.
This velocity field is well-defined, since $u$ is band-limited relative to $\T_u$, and is an infinitesimal flex since $\F(\C;\bK)$ is invariant under the lattice group action $\pi_\T$. Also, $w$ is  factor periodic for $\omega = (\omega_0,\omega_*)$ and the periodic structure $\T'$. It follows that $\Omega(\C)$ contains $(\omega_0,\omega_*)$ for every point $\omega_*$ and so the RUM spectrum contains linear structure of dimension $d-t$.
 \end{proof}

\begin{cor}
Let $\C$ be a crystal framework in $\bR^2$ whose RUM spectrum in $[0,2\pi)^2$ is a proper infinite subset which contains no line segments. Then $\C$ does not possess a crystal flex spanning set which includes a bounded band-limited flex.
\end{cor}

This corollary applies in particular to the 2D zeolite framework $\C_{\rm oct}$ whose motif has bars belonging to a regular octagonal ring of equilateral triangles. Indeed the RUM spectrum of $\C_{\rm oct}$ has been shown to be an
curve containing no line segments \cite{bad-kit-pow,pow-poly}.

The same argument in the proof of Theorem \ref{t:omegalinear} gives a parallel corollary for the geometric flex spectrum. Indeed, let us say that the geometric flex spectrum $\Gamma(\C)$ \emph{contains linear structure} if there exists a point $ \omega = (r_1e^{i\eta_1},\dots ,r_de^{i\eta_d})
$ in $\Gamma(\C)$ such that the intersection of
$\Gamma(\C)$ with the torus
\[
\bT^d_\omega := \{(r_1e^{i\theta_1},\dots ,r_de^{i\theta_d}):\theta_1,\ldots,\theta_d\in[0,2\pi)\}
\]
contains $t$-dimensional linear structure in the sense above for $\bT^d$. Then if there is a crystal spanning set for $\C$ which contains a band-limited flex it follows, as in the proof of Theorem \ref{t:omegalinear}, that $\Gamma(\C)$ contains linear structure.

We can use this observation together with Theorem \ref{thm:crystal} to obtain the following necessary condition for the existence of a crystal flex spanning set or basis. 

\begin{thm}\label{t:gammalinear}
Let $\C$ be a crystal framework in $\bR^d$ whose geometric flex spectrum is a proper infinite subset which contains no linear structure. Then $\C$ does not possess a crystal flex spanning set.
\end{thm}

\begin{proof}
A finite set of geometric flexes $u^{1},\ldots,u^{r}$ with distinct periodicity multifactors $\omega^{1},\ldots,\omega^{r}$ is linearly independent and so, by the hypotheses, the infinitesimal flex space of $\C$ is infinite dimensional. By Theorem \ref{thm:crystal}, any crystal flex spanning set for $\C$ contains a band-limited flex.  By the discussion above it follows that the geometric flex spectrum contains linear structure, a contradiction. 
\end{proof}

\begin{rem}
We recall that early experimental and computational studies of rigid unit modes in silicates typically revealed linear structure in the RUM spectrum for their high temperature phases. See, for example, Dove et al \cite{dov-hei-ham} where studies of curved RUM surfaces, occurring in tridymite for example, are also indicated. 

It is natural to pose the following questions regarding converse implications to the statements in the results above. If the geometric flex spectrum (resp.~RUM spectrum) has linear structure  does it follow that there exists a band-limited flex (resp.~bounded band-limited flex)? 
More generally, it would be of interest to obtain sufficient conditions, including linear structure conditions on the geometric flex spectrum, for the existence of a crystal flex basis or an essential crystal flex basis.  
Such sufficient conditions would explain more fully the nature of the experimental phenomenon in terms of the existence and nonexistence of flexes which are band-limited with respect to lines and planes in the crystal structure. 
\end{rem}

\section{Examples}\label{s:examples}
We now determine crystal flex spanning sets and bases for a number of elementary crystal frameworks. In these examples the spanning sets are finitely generated by a small number of band-limited flexes. To quantify this we introduce the following associated measure of flex complexity for any crystal framework $\C$.

\begin{defn} 
The \emph{flex complexity} $cpx(\C)$ is  the minimum
of the cardinalities of the generating sets for a crystal spanning set for $\F(\C;\bR)$. Moreover, $cpx(\C)=\infty$ if there is no such generating set. 
\end{defn}

\subsection{The frameworks   $\C_{\bZ^2}$ and $\C_{\rm hex}$.} 
Let $\C_{\bZ^2}$ be the grid framework in two dimensions, that is, the framework with joints located on the $\bZ^2$ lattice and bars between nearest neighbours. Let $\S_{\rm grid} = \{u_n, v_n: n\in \bZ\}$ be a set of velocity fields with $u_n$ (resp. $v_m$)  supported by the joints on the line $y=n$ (resp. $x=m$) and with unit velocities in the direction of the support line. Then it is elementary to check, by an exhaustion argument in the style of the proofs below, that $\S_{\rm grid}$ is a  crystal basis for  the space $\F(\C_{\bZ^2};\bR)$.

Let $\C_{\rm hex}$ be the honeycomb framework associated with the regular hexagonal tiling of the plane. Note that any hexagon ring subframework, $\H$ say,   is the support of a  (normalised) local infinitesimal flex, $u_\H$ say, which acts as infinitesimal rotation of $\H$. Let $\S_{\rm hex}$ be a set of identical non-zero flexes of this type for all the honeycomb cells. We claim this is a crystal flex spanning set.  We give the argument for this since it is typical of the simple exhaustion argument needed to show that a given set is a free spanning set or a free basis. Note also that $\sum_\H u_\H$ is the zero infinitesimal flex of 
$\C_{\rm hex}$ so $\S_{\rm hex}$ is not a free basis.
On the other hand, as the proof below shows, it has a curious minimal redundancy property, in the sense that the removal of any flex from $\S_{\rm hex}$
gives a free basis, although not a crystal basis.

From these observations and Proposition \ref{p:hexspanningset} it follows that  $cpx(\C_{\bZ^2})=cpx(\C_{\rm hex}) 
=1$.
\begin{center}
\begin{figure}[ht]
\centering
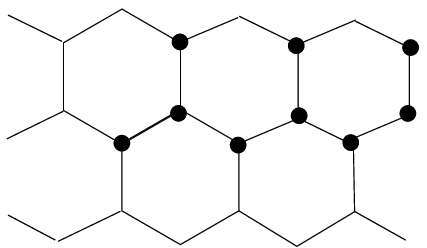
\caption{Part of $\C_{\rm hex}$.}
\label{f:hexflex}
\end{figure}
\end{center}

\begin{prop}\label{p:hexspanningset}
$\S_{\rm hex}$ is a crystal spanning set for the space $\F(\C_{\rm hex};\bR)$ of real infinitesimal flexes of $\C_{\rm hex}$.
\end{prop}

\begin{proof}
Let $z$ be an infinitesimal flex of  $\C_{\rm hex}$. Subtracting a linear combination of the local flexes for the cells $A$ and $B$ indicated in Figure \ref{f:hexflex}, we may assume that the velocity field for $z$ at the origin is $0$. Consider the infinite path of framework points to the right of the origin, lying on cells $1, 2, 3, \dots$, as indicated in Figure  \ref{f:hexflex}.
We may subtract a multiple of the local flex for cell $1$ to ``fix" the first joint, that is, to create a flex with zero velocities at $O$ and at the first joint of the path. We may continue similarly to see that there is an infinite linear combination $w$ of the local flexes for cells $1, 2, 3,\ldots $ such that $z-w$ has zero velocity at each joint of the infinite path.

In the same manner we can subtract an infinite linear combination of the local flexes to the left of cells $A$ and $B$ to obtain a new flex which fixes all the framework points on the two-way infinite horizontal path, $\pi$ say,  through $O$. This may be achieved without making use of
the local flex $u_C$ for the other hexagon incident to the origin. 
We may assume then that $z$
has zero velocities on the joints on $\pi$. The line of joints on $y=1$ may be now be ``fixed" by subtraction of a unique infinite linear combination of the local flexes for the next horizontal line of cells, $a, b, c, \ldots $ etc. At this point the next horizontal line of joints is necessarily fixed by $z$, in view of the flex condition. Continuing this process with the horizontal hexagonal strips above and below $\pi$ we see that there is an infinite linear combination of the local flexes which is equal to the original flex. Also the coefficients of this infinite linear combination are determined uniquely, with the proviso that the local flex $u_C$ is not used in the representation.
It follows that  $\S_{\rm hex}$ is a crystal flex spanning set.
\end{proof}

\subsection{The 2D kagome framework}
We next consider the kagome framework in two dimensions, part of which is indicated in Figure \ref{f:kagbasisu}.
Let $a,b,c$ be the vertices of a triangular subframework of $\C_{\rm kag}$, with horizontal base edge $[p_a,p_b]$,
and let $\L_u^0$ be the (infinite) linear subframework  which contains this edge.
Note that there is an evident one-dimensional subspace of infinitesimal flexes of $\C_{\rm kag}$
that are supported on this linear subframework. Consider the nonzero element $u_0$  in this space which acts on alternate joints of $\L_u^0$ with unit norm velocity fields,   with
\[
u_0(p_a)=(\cos \pi/6, -\sin\pi/6),\,\,\,\,\, u_0(p_b)= (\cos \pi/6, \sin\pi/6).
\]
Let $u_n$, where $n\in \bZ$, be the parallel translates of $u$, naturally labelled, with $u_1$ supported by the first linear subframework $\L_u^1$ above $\L_u^0$. Also, let $\{v_n:n\in \bZ\}$ (resp. $\{w_n:n\in \bZ\}$) be obtained from
 $\{u_n:n\in \bZ\}$ by rotation about the centroid of the triangle $abc$ by $2\pi/3$ (resp. by $4\pi/3$).
 
\begin{center}
\begin{figure}[ht]
\centering
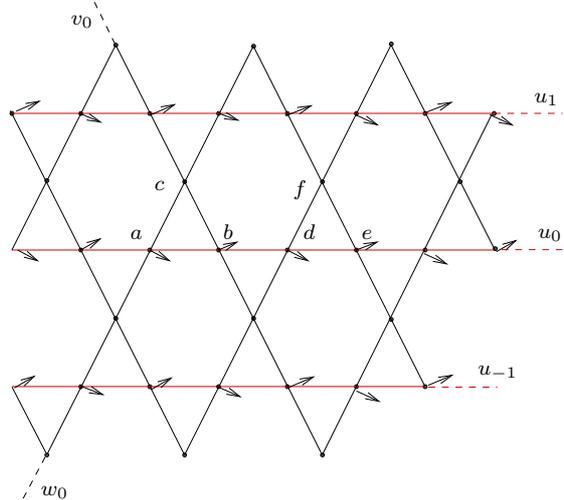
\caption{The horizontal band limited flexes $u_{n}$ for $n \in \bZ$.}
\label{f:kagbasisu}
\end{figure}
\end{center}

The next theorem is due to A. Sait \cite{sai} from which it follows that $cpx(\C_{\rm kag})=1$.

\begin{thm}\label{t:kagomefree}
The set $\B_{\rm kag} = \{u_n, v_n, w_n:n\in \bZ\}$ is a crystal basis for $\F(\Ckag;\bR)$.
\end{thm}

\begin{proof} Since the space group acts on the set it will be sufficient to show that $\B_{\rm kag}$ is a free basis. Write  $\L_u^n$ (resp. $\L_v^n, \L_w^n$), for $n \in \bZ$, for the supporting linear subframeworks for $u_n$ (resp. $v_n, w_n).$
Let $z$ be an infinitesimal flex of $\Ckag$. By subtracting an appropriate linear combination of the three infinitesimal flexes $u_0, v_0, w_0$ we may assume that the three velocity vectors $z_a, z_b, z_c$ for the joints $p_a, p_b, p_c$ of a central triangle subframework are zero. Subtracting an appropriate multiple of $w_1$ we may then arrange $z_d=0$, where $d$ is the next vertex in the direction from $a$ to $b$. Following this, subtracting an appropriate multiple of $v_{-1}$, we may arrange $z_e=0$ for the next vertex.
Continuing in this way we obtain  an infinite linear combination
\[
z' = \sum_{n\in \bZ}\beta_nv_n + \gamma_nw_n
\]
such that the infinitesimal flex $z''=z-z'$  imparts only zero velocities to the joints of $\L_u^0$.
From the flex condition and the rigidity of triangles we deduce that the flex velocities are also zero on the apex vertices
for the triangle subframeworks, such as   $def$, which are horixontal translates of $abc$. Now subtract an appropriate
multiple of $u_1$ so that the resulting flex is zero on $\L_u^1$. Continuing upwards in this manner, and similarly downwards, we obtain an infinite sum representation for $z''$ in terms of the
infinitesimal flexes $u_n$, $n\in \bZ$. 
Thus the original flex $z$ is an infinite sum of the basis vectors. Also the representation is unique and so the proof is complete.
\end{proof}

Combining this result with Theorem \ref{thm:crystal}(iii) we obtain the following description of the space of bounded infinitesimal flexes.

\begin{cor}
With $\B_{\rm kag} = \{u_n, v_n, w_n: n\in \bZ\}$ as above,
\[
\F_\infty(\C_{\rm kag};\bR) = \{u\in \F(\C_{\rm kag};\bR): u = \sum_{n\in \bZ}(\alpha_nu_n+\beta_nv_n + \gamma_nw_n) : (\alpha_n), (\beta_n), (\gamma_n)\in \ell^\infty(\bZ)\}.
\]
\end{cor}

\subsection{The octahedron crystal framework $\C_{\rm Oct}$.}\label{ss:octahedron} We now consider some   crystal frameworks in three dimensions. 

Write $\C_{\rm Oct}$ for the crystal framework of corner-connected regular octahedra, in the symmetric placement for which there is one translation class for them.
The construction of a crystal flex basis for  $\C_{\rm Oct}$ may be determined by viewing 
$\C_{\rm Oct}$ as the union of countably many copies of the 2D grid framework whose orientations and joint connections are consistent with
Figure \ref{f:octagonmotifcolouredD}.

\begin{center}
 \begin{figure}[ht]
 \centering
 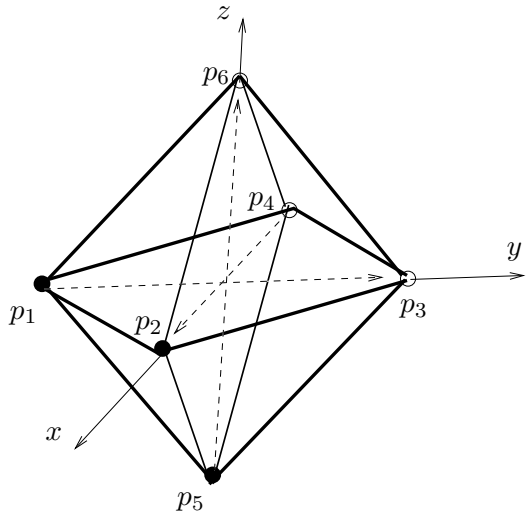
 \caption{A motif of $3$ joints and $12$ bars  for the framework  $\C_{\rm Oct}$.}
 \label{f:octagonmotifcolouredD}
 \end{figure}
 \end{center}

To be precise, assume that the period vectors for  $\C_{\rm Oct}$
are $(2,0,0)$, $(0,2,0)$, $(0,0,2)$ and let $\C_z$ denote the grid framework in the $xy$-plane
which contains the adjacent joints $p_1=(0,-1,0)$, $p_2=(1,0,0)$, $p_3=(0,1,0)$ and $p_4=(-1,0,0)$. These joints lie on a $4$-cycle of bars in the $xy$-plane. Similarly, let  $\C_x$ denote the 2D grid framework in the $yz$-plane
which contains the adjacent joints $p_1$, $p_5=(0,0,-1)$, $p_3$ and $p_6=(0,0,1)$, and let $\C_y$ denote the 2D grid framework parallel to the  $zx$-plane which contains the adjacent joints $p_2$, $p_5$, $p_4$ and $p_6$. 

Let $\C^n_x$ be the translated bar-joint frameworks $\C_x + (2n,0,0),$ for $n\in \bZ$, and similarly define $\C^n_y$ and $\C^n_z$. Then $\C_{\rm Oct}$ is the union of all of these frameworks, that is, it is the bar-joint framework whose joint set is the union of the joints (without multiplicity) and whose set of bars is the union of all the bars.

Let us also define $\C_{\rm Oct}^+$ as the augmented framework in which each regular octahedron is augmented by $3$ bars parallel to the coordinate axes. In view of the infinitesimal rigidity of a convex octahedron it follows that the vector spaces
$\F(\C_{\rm Oct};\bR)$ and $\F(\C_{\rm Oct}^+;\bR)$ are equal.

Let us write $\C_{\rm sq}$ for a  2D bar-joint framework composed of corner connected rigid squares. This is obtained from the 2D grid framework by adding an edge to each alternate square. Let  $\C_{\rm sq}^+$ be the related framework which has both cross diagonals added to the rigid squares.
We may thus view $\C_{\rm Oct}^+$  as the union of copies of  $\C_{\rm sq}^+$, where these copies are augmentation frameworks, $\tilde{\C}^n_x$, $\tilde{\C}^n_y$ and $\tilde{\C}^n_z$ say, of the frameworks ${\C}^n_x$, ${\C}^n_y$ and ${\C}^n_z$.
It follows immediately that the alternation flex $a$ of $\tilde{\C}^n_x$ extends to
a flex $a^{x}_n$ of $\C_{\rm Oct}^+$ which has zero velocities on all the other vertices. Let us similarly define the ``local alternation flexes" $a^{y}_n$ and $a^{z}_n$ for $n\in \bZ$.

Let $r^x, r^y, r^z$ be infinitesimal flexes  for axial rotations   of $\C_{\rm Oct}$ about the principal axes of the central octahedron. These infinitesimal flexes act on the entire framework and are unbounded flexes.  Also we assume the normalisation such that for $\sigma = x,y$ or $z$ the restrictions of  $a^{\sigma}_n$ to the octahedron meeting the $\sigma$-axis agrees with the restriction of $r^\sigma$.

Finally, let $\vec{x}$ be the velocity field in $\F(\C_{\rm Oct};\bR)$ with 
joint velocities $(1,0,0)$ and let $\vec{y}$ and $\vec{z}$ be analogous velocity fields for the $y$ and $z$ directions.

In the next proof we use the following elementary flex projection principle. If the bar $[p_a, p_b]$ lies in a plane $\P$ of  $\bR^3$ and if the joint velocities
$v_a, v_b$ in $\bR^3$ give an infinitesimal flex of the bar $[p_a, p_b]$ then the $\P$ components $v_a', v_b'$ of $v_a, v_b$ also give an infinitesimal flex of the bar. We say that such a flex is an \emph{in-plane} flex when the plane in question is understood.

\begin{thm}\label{t:octahedron} The set $\S$ of velocity fields 
\[
\{r^x, r^y, r^z\} \cup \{\vec{x}, \vec{y}, \vec{z}\}\cup \{a^{x}_n, a^{y}_n, a^{z}_n: n\in \bZ\}
\] is an essential crystal flex  basis for $\F(\C_{\rm Oct};\bR)$. 
\end{thm}

\begin{proof}The subset $\S_0\subseteq \S$ of non-rotational flexes has the crystallographic group action property and so it will suffice to show that $\S$ is a free basis for $\C_{\rm Oct}$. Equivalently, we show that $\S$ is a free basis for  $\C_{\rm Oct}^+$.

Let $z$ be an infinitesimal flex in $\F(\C_{\rm Oct}^+;\bR)$. There is a linear combination
$z_{\rm rig}$ 
of $\vec{x}, \vec{y}, \vec{z}, r^x, r^y, r^z$ which agrees with $z$ on the joints $p_1, \dots , p_6$. Replacing $z$ by $z - z_{\rm rig}$, for some rigid motion infinitesimal flex $z_{\rm rig}$, we may assume that these velocities for $z$ are zero. 

We now make use of the flex projection principle. Note that the velocity field $z_{xy}$ given by the $xy$-plane projection of the joint velocities $z(p)$, for joints in $\tilde{\C}_{z}^0$, is an infinitesimal flex of  $\tilde{\C}_{z}^0$. 
 Since $z$ has zero velocity vectors on the central octahedron it follows that the $xy$-plane projection of $z$ also has zero velocities on the central square of $\tilde{\C}_{z}^0$. It follows that this  in-plane flex is \emph{equal} to the restriction of a scalar multiple of $a_0^z-r^z$. 
In this way we obtain scalar multiples
$\alpha_0 (a_0^z-r^z), \beta_0 (a_0^x-r^x)$ and $\gamma_0  (a_0^y-r^y)$
which provide the in-plane flexes of $z$ for the planes $z=0$, $x=0$ and
$y=0$. 

Consider now the tower subframework given by the tower of octahedra whose connecting joints lie on the $z$-axis. Since $z$ is zero on the central octahedron  supported by $p_1, \dots ,p_6$, denoted $O_{(0,0,0)}$, it follows that the $z$ component of the
joint velocity for a joint on this line is zero. It also follows that
there is an infinitesimal  flex $\beta_0 (a_0^x-r^x) + \gamma_0  (a_0^y-r^y)
$ with joint velocities agreeing with those of $z$ for the joints on the axial line. It follows similarly that there is a flex of the form
\[
w= \alpha (a_0^z-r^z) + \beta_0 (a_0^x-r^x) +\gamma_0  (a_0^y-r^y)
\] 
with this agreement property for the three axial lines through
$O_{(0,0,0)}$.

Replacing $z$ by $z-w$ we may assume that $z$ is zero on $O_{(0,0,0)}$ and on the joints of the three axial lines of $O_{(0,0,0)}$. Note that the restriction of such a flex $z$ to any other octahedron $O$ with an axis on the coordinate axes must be an infinitesimal rotation flex of the octahedron about this axis. Also each such flex of an individual octahedron $O$, on the $\sigma$-axis say, agrees with the restriction of a scalar multiple of  the local alternation flex
$a^{\sigma}_n$,  for some $n \neq 0$. Evidently these infinitesimal flexes act on distinct octahedra  on the axial lines.

It follows that there is an infinite linear combination of these flexes, $w_2$ say, whose restriction to any octahedron on a coordinate axis is equal to the restriction of $z$.
Replacing $z$ by $z-w_2$ we may assume that $z$ is zero on this triple tower, $\T_{Oct}$ say. We now observe that any infinitesimal flex which is zero on $\T_{Oct}$ is the zero flex. This follows from the fact that the entire framework may be built up from $\T_{Oct}$ by attaching octahedra in groups of four such that at each stage every flex which is zero on $\T_{Oct}$ is the zero flex. 
It follows that $z$ must be identically zero. Thus every velocity field $z$ in $\F(\C_{\rm Oct};\bR)$ is an infinite linear combination of the vectors in the set $\S$.

Note that $\S$ is a countable set of velocity fields which tend to zero strictly and is a free spanning set for a vector space of infinitesimal flexes. Also the  scalar coefficients in the identifications above are determined uniquely by
the joint velocities of the flex $z$. Thus, $\S$ is a free infinitesimal flex basis, as required.
\end{proof}

\begin{rem}
From a simple adaptation of the above proof it follows that the space
$\F_\infty(\C_{\rm Oct};\bR)$ of bounded infinitesimal flexes is the space of infinitesimal flexes of the form
\[
 u = \alpha\vec{x}+ \beta\vec{y}+ \gamma\vec{z} + \sum_{n\in \bZ} \alpha_na^{x}_n +\beta_na^{y}_n+\gamma_n a^{z}_n, \quad
 (\alpha_n), (\beta_n), (\gamma_n) \in \ell^\infty(\bZ)
\]

The octahedral framework serves to model the rigid unit atomic structure of the crystal perovskite. Its zero mode (or RUM mode) phonon spectrum  was among early examples  to be computed by  experiment and simulation. Other examples were  quartz and cristobalite. See Giddy et al \cite{gid-et-al}. The RUM spectrum for crystobalite has complete linear structure in $[0, 2\pi)^3$, being  the union of the three line segments $(t,0,0)$, $(0,t,0)$, $(0,0,t)$, for $0\leq t <2\pi$. The existence of this linear structure in the RUM spectrum also follows from the band-limited flexes in $\S$.
\end{rem}

\subsection{The 3D kagome framework}
The kagome framework in $3$ dimensions, $\C_{\rm Kag}$, has the structure of the kagome net, a network of regular tetrahedral units connected pairwise at their vertices. It may be constructed from the 2D kagome
net by first completing its triangles to tetrahedra with alternating up and down orientations to create a horizontal layer framework, and then joining translational copies of such layers. The period vectors  determine a parallelepiped unit which is occupied by a single tetrahedron  and one can determine a motif for $\C_{\rm Kag}$ with a corresponding 12-by-12 matrix-valued symbol function on the 3-torus \cite{owe-pow-crystal}. However the simple layer structure description, together with the symmetry of $\C_{\rm Kag}$, are sufficient to determine a crystal basis by an exhaustion argument, as before, as we now sketch.

For $m\in \bZ$, let $\{u_n^m, v_n^m, w_n^m:n\in \bZ\}$  be the crystal flex bases for the in-plane infinitesimal flex space of the $m^{th}$ layer and note that the infinitesimal flexes in these sets extend (by zero velocity specifications) to infinitesimal flexes of $\C_{\rm Kag}$. If $w$ is an arbitrary infinitesimal flex of $\C_{\rm Kag}$ we may subtract an appropriate infinite linear combination of these flexes to replace $w$ by a flex $w_1$ with the property that the velocity vectors at each joint in the horizontal layers is in the vertical direction. Let us make a distinction between the ``thick" horizontal layer frameworks and their ``thin" kagome subframeworks (which support the added tetrahedra). For layer $m=0$, fix a triangle subframework in the thin layer to be a base triangle. Its tetrahedron has 3 additional bars. Each of these bars
indicates the direction of a linearly localised flex, each of which is a space group element image of $u_0^0$. We may choose a linear combination of these flexes and subtract them from $w_1$ to obtain a flex $w_2$ so that the tetrahedron has zero velocities (under $w_2$) at all 4 of its joints.

Continuing in this way we see that, as in the 2D kagome, there is a crystal flex basis with a single generator and it consists of the space group element images of $u_0^0$. In particular, $cpx(\C_{\rm Kag}) = 1$.

\subsection{The bipyramid framework $\C_{\rm Bipyr}$}
We next define the bipyramid crystal framework $\C_{\rm Bipyr}$ and compute the transfer function associated with the natural primitive periodic structure. An appropriate basis of period vectors takes the form
\[
(1,0,0),\,\,\, (1/2,\sqrt{3}/2,0), \,\,\, (0,0,2h)
\]
and we consider a motif $(F_v, F_e)$ where  $F_v$ consists of the two joints $p_1=(0,0,0), p_2= (1/2,\alpha,-h)$ and $F_e$ consists of the $9$ bars $p_ip_j$ associated with the $9$ directed bars indicated in Figure \ref{f:bipyramidVectors}. Here $h=\sqrt{2/3}$ is the height of a regular tetrahedron of unit sidelength and $\alpha = \sqrt{3}/6$ is the distance from $B$ to $C$, that is, the distance of the centroid of a unit sidelength equilateral triangle to the triangle boundary.

\begin{center}
\begin{figure}[ht]
\centering
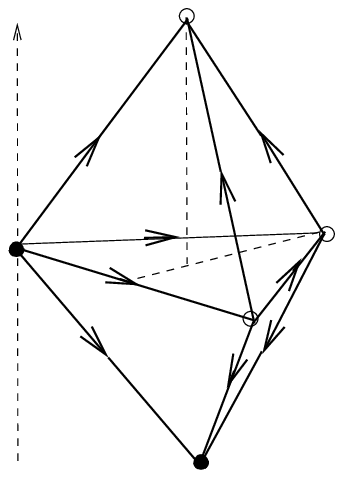
\caption{}
\label{f:bipyramidVectors}
\end{figure}
\end{center}

The following data for the motif edges  will suffice for the computation of the associated transfer function $\Psi_{\rm Bipyr}(z)$.
{\small
\begin{center}
\begin{tabular}{|c|c|c|c|}
\hline
$i$&$p(e_i)$&$k=(k_1,k_2,k_3)$ &$l= (l_1,l_2,l_3)$ \\
\hline
$1$&$(-1,0,0)$&$(0,0,0)$ &$(1,0,0)$ \\
\hline
$2$&$(1/2,-\sqrt{3}/2,0)$&$(1,0,0)$ &$(0,1,0)$ \\
\hline
$3$&$(-1/2,-\sqrt{3}/2,0)$&$(0,0,0)$ &$(0,1,0)$ \\
\hline
$4$&$(-1/2,-\alpha,h)$&$(0,0,0)$ &$(0,0,0)$ \\
\hline
$5$&$(1/2,-\alpha,h)$&$(1,0,0)$ &$(0,0,0)$ \\
\hline
$6$&$(0,2\alpha,h)$&$(0,1,0)$ &$(0,0,0)$ \\
\hline
$7$&$(-1/2,-\alpha,-h)$&$(0,0,0)$ &$(0,0,1)$ \\
\hline
$8$&$(1/2,-\alpha,-h)$&$(1,0,0)$ &$(0,0,1)$ \\
\hline
$9$&$(0,2\alpha,-h)$&$(0,1,0)$ &$(0,0,1)$ \\
\hline
\end{tabular}\\
\medskip
\end{center}
}
A factor-periodic infinitesimal flex, with factor  $\omega=(\omega_1,\omega_2, \omega_3)\in\bC_*^3$, 
has the form
\[
u = (u_{\kappa,k}) = (\omega^k{a}) = (\omega_1^{k_1}\omega_2^{k_2}\omega_3^{k_3}a), \quad k\in \bZ^3,
\]
where
$a=(u_{1,x}, u_{1,y}, u_{1,z}, u_{2,x}, u_{2,y}, u_{2,z}) \in \bR^6$
is the velocity field for the motif joints, being the list of the coordinate velocities
of $u$ at the joints $p_1 $ and $p_2$. 
We have,
\[
\Psi_{\rm Bipyr}({z}^{-1})=\kbordermatrix{
& v_{1,x} &  v_{1,y} & v_{1,z} & \vrule & v_{2,x} &  v_{2,y}& v_{2,z}\\
e_1 &-(1-z_1) &0 & 0 & \vrule & 0 &  0& 0\\
e_2 & (z_1-z_2)/2& -(z_1-z_2)\sqrt{3}/2& 0 & \vrule & 0 &  0 & 0\\ 
e_3 & (1-z_2)/2 & -(1-z_2)\sqrt{3}/2&0 & \vrule & 0 &  0 & 0\\ 
e_4 & -1/2&-\alpha & h & \vrule & 1/2 &  \alpha&-h  \\
e_5 &z_1/2 &-\alpha z_1 & z_1h & \vrule & -1/2 &  \alpha &-h \\ 
e_6 & 0 &2z_2\alpha &z_2h & \vrule & 0 &  -2\alpha & -h\\ 
e_7 &-1/2  &-\alpha &-h & \vrule & z_3/2&  z_3\alpha &z_3h \\
e_8 & z_1/2 &-z_1\alpha  &-z_1h & \vrule & -z_3/2 &  z_3\alpha &z_3h \\
e_9 & 0 &2z_2\alpha &-z_2h & \vrule & 0 &  -2z_3\alpha & z_3h\\ 
}
\]

The submatrix of $\Psi_{\rm Bipyr}(z^{-1})$ for columns $3,4,5$ is a $9 \times 3$ function matrix. We note that the first three rows are zero. Indeed, the three edges $e_1=(v_1,(0,0,0))(v_1,(1,0,0))$, 
$e_2=(v_1,(1,0,0))(v_1,(0,1,0))$ and $e_3=(v_1,(0,0,0))(v_1,(0,0,1))$ are of $v=w$ type and so the entries for columns
$4,5,6$ are zero. Also the entries for column $3$ are zero since the $z$-component is zero for the vectors $p(e_i)$, for $i=1,2,3$. 

We next determine the set of factors $\omega \in \bC_*^3$ for which $\Psi_{\rm Bipyr}(\omega^{-1})a=0$ for some nonzero vector of the form
$a=(0, 0, u_{1,z}, u_{2,x}, u_{2,y}, 0)$.
In doing so we shall determine the factor periodic infinitesimal flexes $u$ with the property that their velocity fields impart only vertical velocities to the joints that lie in the horizontal copies of the fully triangulated framework $\C_{\rm tri}$ and impart only horizontal velocities to the other joints, the polar joints of the constituent bipyramids. In view of the latter condition we refer to these infinitesimal flexes simply as \emph{sheering flexes}.
The required solutions for $\omega$ are the values of $z=(z_1,z_2,z_3)$ for which the submatrix $\Psi_{\rm sub}(z^{-1})$ has nonzero kernel, where
\[\Psi_{\rm sub}(z^{-1})=\kbordermatrix{
 & v_{1,z} & \vrule & v_{2,x} &  v_{2,y}\\
e_4 & h & \vrule & 1/2 &  \alpha\\
e_5 & z_1h & \vrule & -1/2 &  \alpha\\ 
e_6 & z_2h & \vrule & 0 &  -2\alpha\\ 
e_7 & -h & \vrule & z_3/2&  z_3\alpha\\
e_8 & -z_1h & \vrule & -z_3/2 &  z_3\alpha\\
e_9 & -z_2h & \vrule & 0 &  -2z_3\alpha\\ 
}
\]
Noting  the similarity between the rows for $e_6$ and $e_9$ it follows that if $z_3$ is not equal to $-1$ then the kernel is trivial.
On the other hand if $z_3=-1$ then the rows for $e_7, e_8, e_9$ are the negatives of the rows for $e_4, e_5, e_6$. We conclude that $(\omega_1,\omega_2, \omega_3)$ is a periodicity multifactor for a sheering flex $u$ 
if and only if $\omega_3=-1$ and the determinant of the $3 \times 3$ submatrix for the first $3$ rows is zero at $\omega_1, \omega_2$. This determinant is $\alpha h(1 +z_1 +z_2)$. 
We conclude that when $z_1$ and $z_2$ are unimodular then there are exactly two solutions and so the RUM spectrum
$\Omega(\C_{\rm Bipyr})$ contains the set
\[
\{(1,1,1),\, (e^{2\pi i/3},e^{4\pi i/3},e^{\pi i}),\,
(e^{4\pi i/3},e^{2\pi i/3},e^{\pi i})\}.
\]

We also note that 
there are unbounded geometric flexes associated with the solutions $(\omega_1, \omega_2)$ of the equation of $1 +z_1 +z_2=0$, where
$\omega_1 = r$, $\omega_2 = -(1+r)$, and $0<r<1$. These infinitesimal flexes bear some analogy with the unbounded flex of the kite framework considered in Example \ref{eg:kite}. Since there are
uncountably many linearly independent unbounded flexes of this type the infinitesimal flex space $\F(\C_{\rm Bipyr}; \bR)$ is infinite dimensional.

\subsection{Two crystal extensions of $\C_{\rm Bipyr}$}
Consider the crystal framework  $\C_{\rm Bipyr}^e$ which is obtained from  $\C_{\rm Bipyr}$ by adding bonds of length $1$ between polar joints whenever this is possible. 
It is straightforward to see that  $\C_{\rm Bipyr}^e$ is sequentially infinitesimally rigid (\cite{kit-pow}, \cite{owe-pow-crystal}) and so is infinitesimally rigid in the strongest possible sense.
Note that every 2D subframework
parallel to the $xy$-plane is a copy of the fully triangulated framework $\C_{\rm tri}$. 
Let  $\C_{\rm Bipyr}^+$ be obtained from  $\C_{\rm Bipyr}$ by   adding edges and vertices so that \emph{every} triangle in the horizontal copies of $\C_{\rm tri}$ in $\C_{\rm Bipyr}$ is the equator of a bipyramid. This framework may be described as having horizontal layers of maximally packed bipyramids.
The joints are once again of two types, with polar joints having degree $6$, as before, and equatorial joints having degree $18$. One can observe that in fact the two sheering RUMs of $\C_{\rm Bipyr}$
extend to this framework and so, despite the edge rich structure, $\C_{\rm Bipyr}^+$ is not infinitesimally rigid.

\bibliographystyle{abbrv}
\def\lfhook#1{\setbox0=\hbox{#1}{\ooalign{\hidewidth
  \lower1.5ex\hbox{'}\hidewidth\crcr\unhbox0}}}

\end{document}